\documentclass[a4paper,UKenglish,cleveref, autoref,numberwithinsect,nameinlink]{lipics-v2021}

\nolinenumbers

\usepackage[utf8]{inputenc}
\usepackage{amsmath,amssymb}
\usepackage{amsthm}

\usepackage{comment}
\usepackage{hyphenat}
\usepackage{paralist}
\usepackage[numbers,square,sort]{natbib}
\usepackage{xcolor}
\usepackage{nicefrac}
\usepackage{tabularx}
\usepackage{booktabs}
\usepackage{dsfont}

\hideLIPIcs

\usepackage{hyperref}

\usepackage{tikz}
\usetikzlibrary{calc}

\usepackage{pgfplots}

\pgfplotsset{compat=1.10}

\usepgfplotslibrary{fillbetween}

\theoremstyle{plain}
\theoremstyle{definition}

\newtheorem{rrule}[theorem]{Reduction Rule}
\crefname{rrule}{Reduction Rule}{Reduction Rules}
\crefname{observation}{Observation}{Observations}
\newtheorem{hypothesis}{Hypothesis}
\crefname{hypothesis}{Hypothesis}{Hypotheses}

\Crefname{rrule}{Reduction Rule}{Reduction Rules}
\Crefname{corollary}{Corollary}{Corollaries}
\Crefname{proposition}{Proposition}{Propositions}

\crefname{claim}{claim}{claims}
\Crefname{claim}{Claim}{Claims}

\crefname{enumi}{Property}{Properties}
\Crefname{enumi}{Property}{Properties}

\newcommand{\reduced}{\operatorname{red}}

\newcommand{\nt}{\textsc{Negative Triangle}}
\newcommand{\OV}{\textsc{Orthogonal Vectors}}
\newcommand{\kernelsizeconstant}{\beta}
\newcommand{\samplesizeexponent}{\epsilon}

\newcommand{\kernelcomputationexponent}{{\gamma}}

\DeclareMathOperator{\dist}{dist}

\newcommand{\yes}{``Yes''}
\newcommand{\no}{``No''}

\newcommand{\N}{\ensuremath{\mathds{N}}}

\newcommand{\cocl}[1]{\ensuremath{\operatorname{#1}}}
\newcommand{\NP}{\cocl{NP}}
\newcommand{\coNP}{\cocl{coNP}}
\newcommand{\poly}{\cocl{poly}}

\newcommand{\NPnotincoNPslashpoly}{\ensuremath{\NP\not\subseteq \coNP/\poly}}

\newcommand{\vareps}{\varepsilon}

\newcommand{\LCS}{\textsc{Longest Common Subsequence}}

\newcommand{\optalgexponent}{\alpha}

\newcommand{\parameter}{\ell}

\newcommand{\comptimeexponent}{\nu}

\newcommand{\compparameter}{\mu}
\newcommand{\numinstances}{\lambda}

\newcommand{\mcconv}{\textsc{Multicolored Convolution 3-Sum}}

\newcommand{\decprob}[3]{
	\begin{center}
		\begin{minipage}{0.96\linewidth}
			\noindent
			\textsc{#1}
			\begin{compactdesc}
			 \item[\textbf{Input:}]  #2
			 \item[\textbf{Question:}]  #3
			\end{compactdesc}
		\end{minipage}
	\end{center}
}

\usepackage{xargs}

\usepackage{mathtools}

\usepackage{etoolbox}

\title{Parameterized Lower Bounds for Problems in P via Fine-Grained Cross-Compositions}

\author{Klaus Heeger}{Technische Universität Berlin, Algorithmics and Computational Complexity, Germany}{heeger@tu-berlin.de}{https://orcid.org/0000-0001-8779-0890}{Supported by DFG project NI 369/16 ``FPTinP''.}

\author{André Nichterlein}{Technische Universität Berlin, Algorithmics and Computational Complexity, Germany}{andre.nichterlein@tu-berlin.de}{https://orcid.org/0000-0001-7451-9401}{}

\author{Rolf Niedermeier}{Technische Universität Berlin, Algorithmics and Computational Complexity, Germany}{}{https://orcid.org/0000-0003-1703-1236}{}

\authorrunning{Klaus Heeger, André Nichterlein, Rolf Niedermeier}



\Copyright{Klaus Heeger, André Nichterlein, and Rolf Niedermeier}

\acknowledgements{In memory of Rolf Niedermeier, our colleague, friend, and mentor, who sadly passed away before this paper was finished.\\
We thank the anonymous reviewers for their thoughtful and constructive feedback.}

\begin{document}

\maketitle

\begin{abstract}

	We provide a general framework to exclude parameterized running times of the form~$O(\ell^\beta+ n^\gamma)$ for problems that have polynomial running time lower bounds under hypotheses from fine-grained complexity.
	Our framework is based on cross-compositions from parameterized complexity. 
	We (conditionally) exclude running times of the form $O(\ell^{{\gamma}/{(\gamma-1)} - \vareps} + n^\gamma)$ for any $1 < \gamma < 2$ and $\vareps > 0$ for the following problems:
	\begin{itemize}
		\item \textsc{Longest Common (Increasing) Subsequence}: Given two length-$n$ strings over an alphabet~$\Sigma$ (over~$\N$) and~$\ell \in \N$, is there a common (increasing) subsequence of length~$\ell$ in both strings?
		\item \textsc{Discrete Fr\'echet Distance}: Given two lists of~$n$ points each and~$k \in N$, is the Fr\'echet distance of the lists at most~$k$? Here $\ell$ is the maximum number of points which one list is ahead of the other list in an optimum traversal. 
		\item \textsc{Planar Motion Planning}: Given a set of~$n$ non-intersecting axis-parallel line segment obstacles in the plane and a line segment robot (called rod), can the rod be moved to a specified target without touching any obstacles?
		Here~$\ell$ is the maximum number of segments any segment has in its vicinity.
	\end{itemize}

	\noindent Moreover, we exclude running times $O(\ell^{{2\gamma}/{(\gamma -1)}-\vareps} + n^\gamma)$ for any $1 < \gamma < 3$ and $\vareps > 0$ for:
	\begin{itemize}
		\item \textsc{Negative Triangle}: Given an edge-weighted graph with $n$ vertices, is there a triangle whose sum of edge-weights is negative?
		Here~$\ell$ is the order of a maximum connected component.
		\item \textsc{Triangle Collection}: Given a vertex-colored graph with $n$ vertices, is there for each triple of colors a triangle whose vertices have these three colors?
		Here~$\ell$ is the order of a maximum connected component.
		\item \textsc{2nd Shortest Path}: Given an $n$-vertex edge-weighted directed graph, two vertices~$s$ and~$t$, and $k \in \N$, has the second longest $s$-$t$-path length at most~$k$?
		Here~$\ell$ is the directed feedback vertex set number.
	\end{itemize}
	
	Except for \textsc{2nd Shortest Path} all these running time bounds are tight, that is, algorithms with running time $O(\ell^{{\gamma}/{(\gamma-1)}} + n^\gamma )$ for any $1 < \gamma < 2$ and $O(\ell^{{2\gamma}/{(\gamma -1)}} + n^\gamma)$ for any $1 < \gamma < 3$, respectively, are known.
	Our running time lower bounds also imply lower bounds on kernelization algorithms for these problems.
\end{abstract}

\newpage

\section{Introduction}

In recent years, many results in Fine-Grained Complexity showed that many decade-old textbook algorithms for polynomial-time solvable problems are essentially optimal:
Consider as an example \textsc{Longest Common Subsequence} (LCS) where, given two input strings with~$n$ characters each, the task is to find a longest string that appears as subsequence in both input strings.
The classic $O(n^2)$-time algorithm is often taught in introductory courses to dynamic programming \cite{CLRS09}.
\citet{BK15} and \citet{ABW15} independently showed that an algorithm solving LCS in~$O(n^{2-\vareps})$ time for any~$\vareps > 0$ would refute the Strong Exponential Time Hypothesis (SETH).
Such conditional lower bounds have been shown for many polynomial-time solvable problems in the recent years \cite{VW18b}.

One approach to circumvent such lower bounds is ``FPT in P'' \cite{GMN17,AWW16}.
For \LCS{} there is a (quite old) parameterized algorithm running in $O(kn + n \log n)$ time, where $k$ is the length of the longest common subsequence~\cite{Hir77}.
Thus, if~$k$ is small (e.\,g.~$O(n^{0.99})$), then the~$O(n^2)$ barrier can be broken (without refuting the SETH).
A natural question is whether we can do better.
As~$k \le n$, an algorithm running in~$O(k^{1-\vareps}n)$ time for any~$\vareps > 0$ would break the SETH.
However, there are no obvious arguments excluding a running time of~$O(k^2 + n)$.
In fact, such additive running times are not only desirable (as again, for small~$k$ this would be faster than even~$O(kn)$) but also quite common in parameterized algorithmics by employing kernelization:
For \LCS{} the question would be whether there are linear-time applicable data reduction rules that shrink the input to size~$O(k)$.
Then we could simply apply the textbook algorithm to solve \LCS{} in overall~$O(k^2 + n)$ time.
Kernelization is well-studied in the parameterized community~\cite{FLSZ19,ALMNSS20} and also effective in practice for polynomial-time solvable problems such as \textsc{Maximum Matching}~\cite{KKNNZ21} or \textsc{Minimum Cut}~\cite{HN0S20}.

\citet{BK18} showed in an extensive study that such an $O(k^2 + n)$-time algorithm (and indeed many other parameterized algorithms for \LCS{}) would refute the SETH.
This also implies that no such kernelization algorithm as mentioned above is likely to exist.
The results of \citet{BK18} are based on very carefully crafted reductions.

In this work, we follow a different route to obtain similar results for several problems.
We provide an easy-to-apply, general framework to (conditionally) exclude algorithms with running time~$O(k^\beta+ n^\gamma)$ for problems admitting conditional running time lower bounds.
Indeed we show for various string (including \LCS{}) and graph problems as well as problems from computational geometry tight trade-offs between~$\beta$ and~$\gamma$.
This shows that the trivial trade-offs are often the best one can hope for.

\subsection{Related work}

Fine-grained complexity is an active field of research with hundreds of papers.
We refer to the survey of \citet{VW18b} for an overview of the results and employed hypotheses.

Over the last couple of years there has been a lot of work in the direction of ``FPT in P'' for various problems such as \textsc{Maximum Matching}~\cite{KKNNZ21,MNN20,HK19,FLPSW18,IOO18,KN18,CDP19,Duc21}, \textsc{Hyperbolicity} \cite{FKMNNT19,CDP19}, and \textsc{Diameter} \cite{AWW16,CDP19}.
Parameterized lower bounds are rare in this line of work.
Certain linear-time reductions can be used to exclude any kind of meaningful FPT-running times; this is also known as General-Problem-Hardness~\cite{BFNN19}.
Using various carefully crafted reductions, Bringmann and Künnemann~\cite{BK18} show parameterized running time lower bounds (under SETH) for \LCS{} with respect to seven different parameters.
In a similar fashion, Duraj et al.~\cite{DBLP:journals/algorithmica/DurajKP19} show that solving \textsc{Longest Common Increasing Subsequence} in $O( (n \ell)^{1-\epsilon})$ time where $\ell $ is the solution size for some $\epsilon>0$ would refute SETH.

\citet{FMN18} provide lower bounds for \emph{strict} kernelization (i.\,e.\ kernels where the parameter is not allowed to increase) for subgraph detection problems such as \textsc{Negative Weight Triangle} and \textsc{Triangle Collection}.
Conceptually, they use the \emph{diminisher}-framework~\cite{CFM11,FFHKMN20} which was originally developed to exclude polynomial-size strict kernels under the assumption~P${}\ne{}$NP.
The basic idea is to iteratively apply a diminisher (an algorithm that reduces the parameter at a cost of increasing the instance size) and an (assumed) strict kernel (to shrink and control the instance size) to an instance~$I$ of an NP-hard problem.
After a polynomial number of rounds, this overall polynomial-time algorithm will return a constant size instance which is equivalent to~$I$, thus arriving at~P${}={}$NP.
\citet{FMN18} applied the same idea to polynomial-time solvable problems.
In contrast, we rely and adjust the composition-framework by \citet{BDFH09} which was developed to exclude (general) polynomial-size kernels under the stronger assumption~\NPnotincoNPslashpoly.

The composition framework works as follows. 
Consider the example of the NP-hard problem \textsc{Negative-Weight Clique}: Given an edge-weighted graph~$G$ and an integer~$k$, does~$G$ contain a negative-weight $k$-clique, that is, a clique on~$k$ vertices where the sum of the edge-weights of the edges within the clique is negative.

Let~$(G_1, k), (G_2, k), \ldots, (G_t, k)$ be several instances of \textsc{Negative-Weight Clique} with the same~$k$.
Clearly, the graph~$G$ obtained by taking the disjoint union of all~$G_i$ contains a negative-weight $k$-clique if and only if some~$G_i$ contains a negative-weight $k$-clique.
Moreover, the largest connected component of~$G$ has order~$\max_{i \in [t]}\{|V(G_i)|\}$.
Now assume that \textsc{Negative-Weight Clique} has a kernel of size~$O(\ell^c)$ for some constant~$c$ where~$\ell$ is the order of a largest connected component.
By choosing~$t = k^{c+1}$, it follows that kernelizing the instance~$(G,k)$ yields an instance of size less than~$\ell$, that is, less bits than the number of instances encoded in~$G$.
Given the NP-hardness of \textsc{Negative-Weight Clique} such a compression seems challenging; indeed it would imply~$\NP\subseteq \coNP/\poly$~\cite{FS11}, which in turn results in a collapse of the polynomial hierarchy.

Compositions and their extension cross-composition~\cite{BJK14} are extensively employed in the parameterized complexity literature.
Moreover, to exclude kernels whose size is bounded by polynomials of a specific degree adjustments have been made to the composition framework~\cite{DM14}.

\subparagraph{Parameter trade-offs.}
For several of our running time lower bounds we have tight upper bounds that are derived from a simple case distinction argument.

\begin{observation}[folklore]\label{thm:upper-bound}
	If a problem~$\mathcal{P}$ admits an $\widetilde O(\parameter^\beta n^\gamma)$-time algorithm\footnote{The $\widetilde O$ hides polylogarithmic factors.}, then it admits for every $\lambda > 0$ an~$\widetilde O(\parameter^{\beta + \frac{\gamma \cdot \beta}{\lambda}} + n^{\gamma + \lambda})$-time algorithm.
\end{observation}

\begin{proof} 
	If $\parameter \le n^{\frac{\lambda}{\beta}}$, then the $\widetilde O(\parameter^\beta n^\gamma)$-time algorithm runs in $\widetilde O(n^{\gamma +\lambda})$ time.
	Otherwise~$n \le \parameter^{\frac{\beta}{\lambda}}$.
	Then the $\widetilde O(\parameter^\beta n^{\gamma} )$-algorithm then in $\widetilde O (\parameter^{\beta + \frac{\gamma \cdot \beta}{\lambda}})$ time.
\end{proof}

\subsection{Our Results \& Technique} \label{ssec:our-results}

We provide a composition-based framework to establish parameterized running time lower bounds and apply the framework to \textsc{Longest Common Subsequence}, \textsc{Longest Common (Weakly) Increasing Subsequence}, \textsc{Discrete Fr\'echet Distance}, \textsc{Planar Motion Planning}, \textsc{Negative Triangle}, and \textsc{2nd Shortest Path} (see \cref{ssec:prelim} for the problem definitions). 
Using similar ideas we obtain running time lower bounds for \textsc{Triangle Collection}.
For all these problems except \textsc{2nd Shortest Path} parameterized by the directed feedback vertex set there exist matching running time upper bounds.
We refer to \cref{tab:overview} for an overview on the specific results and the parameterization.
\begin{table}[t!]
	\caption{Overview of achievable running times.
	The upper part of the table lists the results for four problems that can be solved in~$O(n^2)$ time but under SETH or 3SUM-hypothesis not in~$O(n^{2-\vareps})$ time for any~$\vareps > 0$.
	The lower part lists results for three graph problems that, based on the APSP-hypothesis, do not admit $O(n^{3-\vareps})$-time algorithms.
	The parameterized upper and lower bounds are visualized in \cref{fig:running-time-bounds}.
	}
	\label{tab:overview}
	\begin{tabularx}{\textwidth}{llX}
		\toprule
		\multirow{4}{*}{\rotatebox[origin=c]{90}{problems}}
		& \multicolumn{2}{l}{\textsc{Longest Common Subsequence} \hfill $\ell \hat{=}{}$solution size} \\ 
		& \multicolumn{2}{l}{\textsc{Longest Common (Weakly) Increasing Subsequence} \hfill $\ell \hat{=}{}$solution size} \\ 
		& \multicolumn{2}{l}{\textsc{Discrete Fr\'echet Distance} \hfill $\ell \hat{=}{}$maximum shift} \\ 
		& \multicolumn{2}{l}{\textsc{Planar Motion Planning} \hfill $\ell \hat{=}{}$max.\ number of segments in the vincinty of any segment}\\ 
		\midrule
		\multirow{5}{*}{\rotatebox[origin=c]{90}{results}}
		& upper bounds		& lower bounds \\
		\cmidrule(lr){2-3}
		& $O(n^2)$ \cite{DBLP:journals/dam/Grabowski16,DBLP:journals/siamcomp/AgarwalAKS14,DBLP:conf/swat/Vegter90,DBLP:conf/isaac/AgrawalG20}  & no $O(n^{2-\vareps})$ assuming SETH / 3SUM \cite{ABW15,DBLP:conf/focs/Bringmann14,BK15,DBLP:journals/comgeo/GajentaanO95} \\
		& $\widetilde O(\ell n)$ \cite{Hir77,DBLP:journals/algorithmica/SifronyS87,DBLP:journals/jda/KutzBKK11} & \\
		& $\widetilde O(\ell^{{\gamma}/{(\gamma-1)}} + n^{\gamma})$ for each $\gamma > 1$ & no $O(\ell^{{\gamma}/{(\gamma-1)} - \vareps} + n^\gamma)$ for any $\gamma < 2$ \cite{BK18,DBLP:journals/algorithmica/DurajKP19} \\
		& (\Cref{thm:upper-bound}) & (\Cref{cor:LCSlower,cor:DFDlower,cor:PMPlower})\\
		\bottomrule
		\toprule
		\multirow{2}{*}{\rotatebox[origin=c]{90}{problems}}
		& \multicolumn{2}{l}{\textsc{Negative Triangle} \hfill $\ell \hat{=}{}$size of maximum component} \\
		& \multicolumn{2}{l}{\textsc{Triangle Collection} \hfill $\ell \hat{=}{}$size of maximum component} \\
		& \multicolumn{2}{l}{\textsc{2nd Shortest Path} (only lower bounds) \hfill $\ell \hat{=}{}$directed feedback vertex set size} \\ 
		\midrule
		\multirow{5}{*}{\rotatebox[origin=c]{90}{results}}
		& upper bounds		& lower bounds \\
		\cmidrule(lr){2-3}
		& $O(n^3/2^{\Theta (\log^{0.5}n)})$ \cite{DBLP:journals/siamcomp/Williams18,DBLP:journals/talg/ChanW21} & no $O(n^{3-\vareps})$ assuming APSP \cite{DBLP:journals/jacm/WilliamsW18}  \\
		& $\widetilde O(\ell^2 n)$ (folklore) \\
		& $\widetilde O(\ell^{{2\gamma}/{(\gamma-1)}} + n^\gamma)$ for each $\gamma > 1$ & no $O(\ell^{{2\gamma}/{(\gamma -1)}-\vareps} + n^\gamma)$ for any $\gamma < 3$ \\
		& (\Cref{thm:upper-bound}) & (\Cref{cor:nt-2SPlower,cor:ntlower,prop:tclower}) \\
		\bottomrule
	\end{tabularx}
\end{table}
Moreover, we visualize in \cref{fig:running-time-bounds} the trade-offs in the running times that are (im-)possible.
\begin{figure}
	\centering
	\begin{tikzpicture} 
		\begin{axis}[
				width=0.4\textwidth,
				height=0.3\textwidth,
				title={$O(n^2)$-solvable},
				xlabel={$\gamma$},ylabel={$\beta$},
				ymax=5,xmax=4,
				ymin=1,xmin=1
			]
			\addplot[name path=f,domain=1:4,green!50!black,samples=100] {x / (x-1)};
			\addplot[color=green!50!black,mark=none] coordinates {(2, 1) (2, 2)};

			\path[name path=axis] (axis cs:0,0) -- (axis cs:10,0);
			\path[name path=axis-top] (axis cs:0,10) -- (axis cs:10,10);

			\addplot [thick,fill=green,fill opacity=0.15] fill between[of=f and axis-top,soft clip={domain=1:4}];
			\addplot [thick,fill=green,fill opacity=0.15] fill between[of=f and axis,soft clip={domain=2:4}];
			\addplot [thick,color=red,fill=red,fill opacity=0.15] fill between[of=f and axis,soft clip={domain=1:2}];
		\end{axis}
	\end{tikzpicture}  \hfil
	\begin{tikzpicture} 
		\begin{axis}[
				width=0.4\textwidth,
				height=0.3\textwidth,
				title={$O(n^3)$-solvable},
				xlabel={$\gamma$},ylabel={$\beta$},
				ymax=5,xmax=4,
				ymin=1,xmin=1
			]
			\addplot[name path=f,domain=1:4,green!50!black,samples=100] {2*x / (x-1)};
			\addplot[color=green!50!black,mark=none] coordinates {(3, 1) (3, 3)};

			\path[name path=axis] (axis cs:0,0) -- (axis cs:10,0);
			\path[name path=axis-top] (axis cs:0,10) -- (axis cs:10,10);

			\addplot [thick,fill=green,fill opacity=0.15] fill between[of=f and axis-top,soft clip={domain=1:4}];
			\addplot [thick,fill=green,fill opacity=0.15] fill between[of=f and axis,soft clip={domain=3:4}];
			\addplot [thick,color=red,fill=red,fill opacity=0.15] fill between[of=f and axis,soft clip={domain=1:3}];
		\end{axis}
	\end{tikzpicture}
	\caption{
		Overview on the possible (in green) and unlikely (in red) trade-offs in running times of the form~$O(n^\gamma + \ell^\beta)$.
		Left: First category for~$O(n^2)$-time solvable problems (upper part in \cref{tab:overview}); right: second category for~$O(n^3)$-time solvable problems (lower part in \cref{tab:overview}).
	}
	\label{fig:running-time-bounds}
\end{figure}
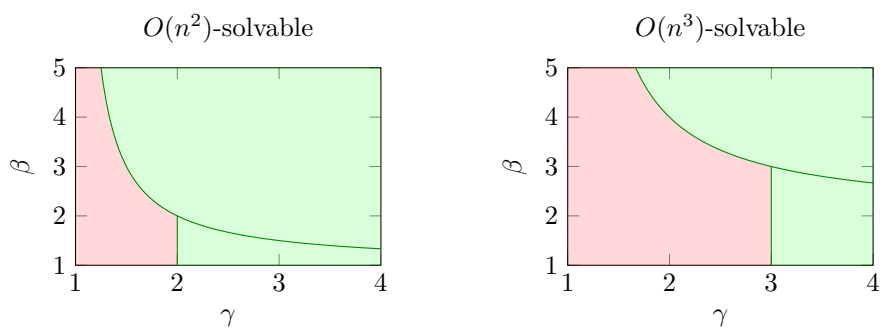

\subparagraph{Framework.} 
We adjust the cross-composition framework to obtain lower bounds for polynomial time solvable problems.
As an example, consider \textsc{Negative-Weight Triangle}, that is, \textsc{Negative-Weight Clique} with~$k$ fixed to three.
Assuming the APSP-hypothesis, \textsc{Negative-Weight Clique} cannot be solved in~$O(n^{3-\vareps})$ time~\cite{DBLP:journals/jacm/WilliamsW18}.
The first difference to the cross-composition framework is that we start with one instance~$G$ of \textsc{Negative-Weight Triangle} which we then decompose into many small instances as follows:
Partition the vertices~$V(G)$ of~$G$ into~$z$ many sets~$V_1, \ldots, V_z$ of size~$n/z$, where~$z$ is chosen depending on the running time we want to exclude (see the proof of \cref{lem:sampling} in \cref{ssec:neg-wei-clique} for the actual formula specifying~$z$). 
Then, we create~$z^3$ instances of~\textsc{Negative-Weight Clique}: for each~$(i,j,k) \in [z]^3$ take the graph~$G[V_i \cup V_j \cup V_k]$.
Clearly, we have that~$G$ contains a negative-weight triangle if and only if at least one of the created instances contains a negative-weight triangle.

Next, we apply the composition as explained above (the disjoint union) for \textsc{Negative-Weight Clique} to obtain an instance~$G'$ with~$n' = z^3 \cdot \nicefrac{n}{z} =  z^2 n$ vertices.
Note that the size~$\ell$ of a largest connected component in~$G'$ is~$3n/z$.
Hence, an algorithm running in time~$O(n^\gamma + \ell^\beta)$ for \textsc{Negative-Weight Triangle} solves~$G'$ in time~$O(z^{2\gamma} n^\gamma + (3n/z)^\beta)$.
By carefully choosing~$z$ as a function in~$n$, $\beta$, and $\gamma$, we get that this is in~$O(n^{3-\vareps})$ for various combinations of~$\gamma$ and~$\beta$.

The property that \textsc{Negative-Weight Triangle} can be decomposed as above is not unique to the problem. 
In fact, this has been observed already:
``Many problems, like SAT, have a simple self-reduction proving that the ``Direct-OR'' version is hard, assuming the problem itself is hard''~\cite{DBLP:journals/talg/AbboudBHS22}.
Our framework formalizes this notion of decomposition (see \cref{sec:framework} for a definition) and adjusts the cross-composition definition.
We furthermore show that commonly used ``hard'' problems such as \textsc{Orthogonal Vectors}, \textsc{3-Sum}, and \textsc{Negative-Weight $k$-Clique} are decomposable.
Thus, it remains to show cross-compositions in order to apply our framework and obtain lower bounds.

\subsection{Preliminaries and Notation.} \label{ssec:prelim} 

\subparagraph*{Problem definitions.} For $\ell \in \N$ we set~$[\ell] := \{1,2,\ldots,\ell\}$.

\decprob{\textsc{Orthogonal Vectors}}
{Two size-$n$ sets~$A, B \subseteq \{0, 1\}^d$ for some~$d \in \N$.}
{Are there~$a \in A $ and $b \in B$ so that $a$ and $b$ are orthogonal, i.e., for each $i \in [d]$, the $i$-th coordinate of~$a$ or the $i$-th coordinate of~$b$ is zero.}

We denote the restriction of \textsc{Orthogonal Vectors} to instances with $d \le O (\log n)$ as \textsc{Orthogonal Vectors}  \emph{with logarithmic dimension}.

\decprob{\textsc{3-Sum}}
{An array~$A$ of $n$ integers.}
{Are there $i, j, h\in [n]$ such that $A[i] + A[j] + A[h] = 0$?}

\decprob{\textsc{Negative-Weight $k$-Clique}}
{An edge-weighted graph~$G$ on~$n$ vertices.}
{Does $G$ contain a $k$-clique of negative weight?}

\decprob{\LCS}
{Two strings~$x^1$ and~$x^2$ of length~$n$ over an alphabet~$\Sigma$ and~$k \in \N$.}
{Decide whether there is a common subsequence of length~$k$ of~$x^1$ and~$x^2$.}

\decprob{\textsc{Longest Common (Weakly) Increasing Subsequence}}
{Two strings~$x^1$ and~$x^2$ of length~$n$ over~$\N$ and~$k \in \N$.}
{Decide whether there is a common subsequence~$y$ of length~$k$ of $x^1$ and $x^2$ with~$y[i] < y[{i+1}]$ ($y[i] \le y[i+1]$) for all~$i\in [k -1]$.}

\decprob{\textsc{Discrete Fr\'echet Distance}}{
Two lists~$P = (p_1, \dots, p_n)$ and $Q = (q_1, \dots, q_n)$ of points in the plane and~$k \in \N$.}{
Is the Fr\'echet distance of $P$ and $Q$ at most~$k$, that is, there are two surjective, non-decreasing functions~$f_P, f_Q : [2n] \rightarrow [n]$ with $f_P (1) = 1 = f_Q (1)$, $f_P (2n) = n = f_Q (2n)$ and $\max_{i \in [2n]}\operatorname{dist} (p_{f_P (i)}, q_{f_Q (i)}) \le k$?}

\decprob{\textsc{Planar Motion Planning}}
{A set of~$n$ non-intersecting, non-touching, axis-parallel line segment obstacles in the plane and a line segment robot (a rod or ladder), a given source, and a given goal.}
{Can the rod be moved (allowing both translation and rotation) from the source to the goal without colliding with the obstacles?}

\decprob{\textsc{2nd Shortest Path}}
{An $n$-vertex edge-weighted directed graph~$G$, vertices $s$ and~$t$, and~$k \in N$.}
{Has the 2nd-shortest $s$-$t$-path length at most~$k$?}

\decprob{\textsc{Triangle Collection}}
{A vertex-colored graph~$G$ on~$n$ vertices.}
{For each combination of three colors, does $G$ contain a triangle whose vertices are colored with three colors?}

\subparagraph{Hypotheses.}
The conditional lower bounds in this work are based on SETH, 3-Sum-, and the APSP-Hypothesis (see \citet{VW18b} for more details).

\begin{hypothesis}[SETH]\label{hyp:SETH}
	For every $\vareps > 0$ there exists a $k \in \N$ such that $k$-SAT cannot be solved in $O(2^{(1-\vareps)n})$ time, where $n$ is the number of variables.
\end{hypothesis}
\begin{hypothesis}[3SUM]\label{hyp:3SUM}
	\textsc{3-Sum} on~$n$ integers in~$\{-n^4,\ldots,n^4\}$ cannot be solved in~$O(n^{2-\vareps})$ time for any $\vareps > 0$.
\end{hypothesis}
\begin{hypothesis}[APSP]\label{hyp:APSP}
	\textsc{All Pairs Shortest Path} on $n$-vertex graphs with polynomial edge weights cannot be solved in~$O(n^{3-\vareps})$ time for any $\vareps > 0$.
\end{hypothesis}

\subparagraph{Parameterized Complexity.}
In many of the above problems, $n$ is \emph{not} the input size but a parameter and the input size is bounded by~$n^{O(1)}$.
A parameterized problem is a set of instances~$(I,p)\in \Sigma^*\times\Sigma^*$, where~$\Sigma$ denotes a finite alphabet, $I$ denotes the classical instance and~$p$ the parameter.
A \emph{kernelization} is a polynomial-time algorithm that maps any instance~$(I,p)$ to an equivalent instance~$(I',p')$ (the \emph{kernel}) such that~$|I'|+p'\leq f(p)$ for some computable function~$f$.
If~$f$ is a polynomial, then~$(I',p')$ is a polynomial-size kernel.

In this work we restrict ourselves to the following:
Either~$p = n$ ($p$ is a single parameter) or~$p = (n, \ell)$ ($p$ is a combined parameter).
Moreover, both $n$ and~$\ell$ are always nonnegative integers, $n$ is related to the input size but~$\ell$ is not ($\ell$ can be seen as ``classical'' parameter).

\section{Framework} \label{sec:framework}

Our framework has the following three steps (see \cref{ssec:our-results} for a high-level description).
\begin{enumerate}
	\item Start with an instance~$(I,n_\mathcal{P})$ of a ``hard'' problem~$\mathcal{P}$ and decompose it into the disjunction of~$t$ instances~$(I_1,n_1), \dots, (I_t,n_t)$ of~$\mathcal{P}$.
	In \cref{sec:decomp}, we provide such decompositions for the frequently used hard problems \textsc{3-Sum}, \textsc{Orthogonal Vectors}, and \textsc{Negative Weight $k$-Clique}.

	\item Compose $(I_1,n_1), \dots, (I_t,n_t)$ into one instance~$(J, (n, \ell))$ of the ``target'' problem using an OR-cross-composition.
	This step has to be done for the application at hand.
 
	\item Apply the assumed $\widetilde{O} (n^{\kernelcomputationexponent} + \parameter^\kernelsizeconstant)$-time algorithm to~$J$.
	If the combination of $\kernelcomputationexponent$ and $\kernelsizeconstant$ is small enough, then the resulting algorithm will be faster than the lower bound for~$\mathcal{P}$.
\end{enumerate}

To give a more formal description of our framework, we first define decompositions and cross-compositions.
Note that all mentioned problems are parameterized problems.

\begin{definition}[OR-decomposition]
	For~$\alpha > 1$ an \emph{$\alpha$-OR-decomposition} for a problem~$\mathcal{P}$ is an algorithm that, given~$\numinstances > 0$ and an instance~$(I,n)$ of~$\mathcal{P}$, computes for some $\alpha' < \alpha$ in $\widetilde{O}(n^{\alpha'})$~time $t \in \widetilde{O} (n^{\alpha \numinstances / (\alpha + \numinstances)})$ many instances~$(I_1,n_1), \ldots, (I_t, n_t)$ of~$\mathcal{P}$ such that 
	\begin{itemize}
		\item $(I,n) \in \mathcal{P}$ if and only if $(I_i,n_i) \in \mathcal{P} $ for some~$i \in [t]$, and
		\item $n_i \in \widetilde{O} (n^{\alpha / (\alpha + \numinstances)})$ for all $i \in [t]$.
	\end{itemize}
\end{definition}

We say a problem~$\mathcal{P}$ is \emph{$\alpha$-OR-decomposable} if there exists an $\alpha$-OR-decomposition for it.
For some problems it is easier to show OR-decomposability than others.
Thus, using appropriate reductions to transfer OR-decomposability can be desirable (we do so in \cref{sec:decomp} when showing that \textsc{3-Sum} is $2$-OR-decomposable).
Quasi-linear time reductions that do not increase the parameter to much are one option.
To this end, we say a reduction that given an instance~$(I^\mathcal{P},n^\mathcal{P})$ of~$\mathcal{P}$ produces an instance~$(I^\mathcal{Q},n^\mathcal{Q})$ of~$\mathcal{Q}$ is \emph{parameter-preserving} if~$n^\mathcal{Q} \in \widetilde{O} (n^\mathcal{P})$. 

\begin{proposition}
	\label{prop:decompose}
	Let $\mathcal{P}$ and $\mathcal{Q}$ be two problems such that there are quasi-linear-time parameter-preserving reductions from~$\mathcal{P}$ to~$\mathcal{Q}$ and from~$\mathcal{Q}$ to $\mathcal{P}$.
	Then for any $\alpha > 1$, $\mathcal{P}$ is $\alpha$-OR-decomposable if and only if $\mathcal{Q}$ is.
\end{proposition}

{
\begin{proof}
	Assume that $\mathcal{P}$ is $\alpha$-OR-decomposable (the case that $\mathcal{Q}$ is $\alpha$-OR-decomposable is symmetric).
	We now give an $\alpha$-OR-decomposition for~$\mathcal{Q}$.
	Given an instance~$(I^{\mathcal Q},n^{\mathcal{Q}})$ and $\lambda > 0$, we first reduce~$(I^{\mathcal{Q}},n^{\mathcal{Q}})$ to an instance~$(I^{\mathcal{P}},n^{\mathcal{P}})$.
	Afterwards, we apply the $\alpha$-OR-decomposition from~$\mathcal{P}$, resulting in instances~$(I_1^{\mathcal{P}},n_1^{\mathcal{P}}), \dots, (I_t^{\mathcal{P}},n_t^{\mathcal{P}})$ of~$\mathcal{P}$.
	Finally, we reduce each instance~$(I_i^{\mathcal{P}},n_i^{\mathcal{P}})$ to an instance~$(I_i^{\mathcal{P}},n_i^{\mathcal{P}})$.
	This clearly is an $\alpha$-OR-decomposition for~$\mathcal{Q}$.
\end{proof}
}

For the second step of our framework, we introduce fine-grained OR-cross-compositions:

\begin{definition}[fine-grained OR-cross-composition]
	For~$\comptimeexponent \ge 1,\compparameter \ge 0$ an \emph{$(\comptimeexponent,\compparameter)$-OR-cross-composition} from a problem~$\mathcal{P}$ to a problem~$\mathcal Q$ is an algorithm~$\mathcal{A}$ which 
	takes as an input~$t$~instances~$(I_1^{\mathcal{P}},n_1^{\mathcal{P}}), \dots, (I_t^{\mathcal{P}},n_t^{\mathcal{P}})$ of $\mathcal{P}$,
	runs in~$\widetilde{O}( t \cdot n_{\max}^\comptimeexponent + \sum_{i=1}^t |I_i^{\mathcal{P}}|)$ time with~$n_{\max} \coloneqq \max_{i \in [t]} n_i^{\mathcal{P}}$, and computes an instance~$(I^{\mathcal{Q}}, (n^{\mathcal{Q}},\parameter^{\mathcal{Q}}))$ of $\mathcal{Q}$ such that
	\begin{enumerate}
		\item $(I^{\mathcal{Q}}, (n^{\mathcal{Q}},\parameter^{\mathcal{Q}})) \in \mathcal{Q}$ if and only if $(I_i^{\mathcal{P}},n_i^{\mathcal{P}}) \in \mathcal{P} $ for some~$i \in [t]$, and
		\item $n^{\mathcal{Q}} \in \widetilde{O}( t \cdot n_{\max}^\comptimeexponent)$ and $\parameter^{\mathcal{Q}} \in \widetilde{O} (n_{\max}^\compparameter)$.
	\end{enumerate}
\end{definition}

We say a problem~$\mathcal{P}$ \emph{$(\comptimeexponent,\compparameter)$-OR-cross-composes} into a problem~$\mathcal{Q}$ if there exists an $(\comptimeexponent,\compparameter)$-OR-cross-composition from~$\mathcal{P}$ to~$\mathcal{Q}$.

\begin{theorem} \label{thm:meta-composition}
	Let~$\alpha > \comptimeexponent \ge 1$, $\kernelcomputationexponent > 1$, and~$\compparameter > 0$ with $\alpha > \comptimeexponent \cdot \kernelcomputationexponent$.
	Let~$\mathcal{P}$ be an $\alpha$-OR-decomposable problem with parameter~$n_\mathcal{P}$ that $(\comptimeexponent,\compparameter)$-OR-cross-composes into a problem~$\mathcal{Q}$ with parameters~$n_\mathcal{Q}$ and~$\parameter_\mathcal{Q}$.
	If there is an~$\widetilde{O} (n_\mathcal{Q}^\kernelcomputationexponent + \parameter_\mathcal{Q}^\kernelsizeconstant)$-time algorithm for~$\mathcal{Q}$ and 
	\[0 < \kernelsizeconstant < \frac{\kernelcomputationexponent \cdot (\optalgexponent - \comptimeexponent)}{(\kernelcomputationexponent - 1) \cdot \compparameter},\]
	then~$\mathcal{P}$ can be solved in~$O(n_\mathcal{P}^{\alpha - \vareps})$ time for some~$\vareps > 0$ time.
\end{theorem}

\begin{proof}
	Let~$(I_{\mathcal{P}},n_{\mathcal{P}})$ be an instance of~$\mathcal{P}$.
	Our algorithm to solve~$(I_{\mathcal{P}},n_{\mathcal{P}})$ runs in the following steps:
	\begin{enumerate}
		\item Apply the~$\alpha$-OR-decomposition (with~$\numinstances$ specified below) to obtain the instances~$(I_1,n_1), \allowbreak\ldots, \allowbreak(I_t,n_t)$ with~$\max_{i \in [t]} n_i \le q \coloneqq n_{\mathcal{P}}^{\alpha / (\alpha + \numinstances)}$ and $t \coloneqq n_{\mathcal{P}}^{\alpha \numinstances / (\alpha + \numinstances)} = q^\numinstances$.
		Note that, by definition, this step runs in $ \widetilde O (n_{\mathcal P}^{\alpha'})$ time for some~${\alpha' < \alpha}$.

		\item Apply the $(\comptimeexponent,\compparameter)$-OR-cross-composition to compute the instance~$(I_{\mathcal{Q}}, (n_{\mathcal{Q}},\ell_{\mathcal{Q}}))$ for~$\mathcal{Q}$ from~$(I_1,n_1), \ldots, (I_t, n_t)$. 
			Note that the running time and~$n_\mathcal{Q}$ is in~$\widetilde{O} (t \cdot q^\comptimeexponent + \sum_{i=1}^t |I_i^{\mathcal{P}}|) = \widetilde{O} (q^{\numinstances + \comptimeexponent} + n_{\mathcal P}^{\alpha'})$.
			Moreover, $\ell_{\mathcal{Q}} \in \widetilde{O} (q^\compparameter)$.
		\item Apply the algorithm with running time~$\widetilde{O} (n_\mathcal{Q}^\kernelcomputationexponent + \parameter_\mathcal{Q}^\kernelsizeconstant)$.
			This requires~$\widetilde{O}(q^{(\numinstances + \comptimeexponent)\kernelcomputationexponent} + q^{\compparameter \cdot \kernelsizeconstant})$ time.
	\end{enumerate}
	It remains to show that all three steps run in~$O(n_\mathcal{P}^{\alpha - \vareps})$ time for some~$\vareps > 0$.
	To this end, we now specify~$\numinstances = \kernelsizeconstant \compparameter / \kernelcomputationexponent - \comptimeexponent$.
	Note that $\numinstances + \comptimeexponent = \kernelsizeconstant \compparameter / \kernelcomputationexponent < \compparameter \cdot \kernelsizeconstant $ and thus it suffices to show that the third step runs in $\widetilde{O} (n_{\mathcal{P}}^{\alpha - \epsilon}) $ for some $\epsilon > 0$.
	The last step runs in~$\widetilde{O}(q^{\compparameter \cdot \kernelsizeconstant}) = \widetilde{O}(n_\mathcal{P}^{\alpha \compparameter \cdot \kernelsizeconstant / (\alpha + \numinstances)})$ time.
	The exponent is 
	\begin{align*}
		\frac{\alpha \compparameter \kernelsizeconstant} {\alpha + \numinstances} 
			& = \frac{\alpha \compparameter \kernelsizeconstant} {\alpha + \kernelsizeconstant \compparameter / \kernelcomputationexponent - \comptimeexponent} 
			  = \frac{\alpha \kernelsizeconstant \kernelcomputationexponent\compparameter} {\kernelcomputationexponent (\alpha - \comptimeexponent) + \kernelsizeconstant \compparameter}
			  = \frac{\alpha \kernelcomputationexponent \compparameter } {\kernelcomputationexponent (\alpha - \comptimeexponent) / \kernelsizeconstant + \compparameter}\,.
	\end{align*}
	By assumption we have~$\kernelsizeconstant < \frac{\kernelcomputationexponent \cdot (\optalgexponent - \comptimeexponent)}{(\kernelcomputationexponent - 1) \cdot \compparameter}$ and thus the exponent is
	\begin{align*}
		\frac{\alpha \kernelcomputationexponent \compparameter } {\kernelcomputationexponent (\alpha - \comptimeexponent) / \kernelsizeconstant + \compparameter}
		< \frac{\alpha \kernelcomputationexponent\compparameter } {\kernelcomputationexponent (\alpha - \comptimeexponent) / \left( \frac{\kernelcomputationexponent \cdot (\optalgexponent - \comptimeexponent) }{(\kernelcomputationexponent - 1) \cdot \compparameter}\right) + \compparameter}
		= \frac{\alpha \kernelcomputationexponent\compparameter } {(\kernelcomputationexponent - 1) \cdot \compparameter + \compparameter}
		= \alpha.
	\end{align*}
	
	There is still one thing left to do:
	We must ensure that $\lambda > 0$.
	This will not always be the case as when $\kernelsizeconstant \rightarrow 0$, $\lambda$ gets negative.
	However, an~$\widetilde{O} (n_\mathcal{Q}^\kernelcomputationexponent + \parameter_\mathcal{Q}^\kernelsizeconstant)$-time algorithm also implies for any~$\kernelsizeconstant' > \kernelsizeconstant$ an~$\widetilde{O} (n_\mathcal{Q}^\kernelcomputationexponent + \parameter_\mathcal{Q}^{\kernelsizeconstant'})$-time algorithm.
	Thus, we can simply pick some larger~$\kernelsizeconstant'$ such that the corresponding $\lambda'$ is larger than 0.
	To do so, let $\kernelsizeconstant_{\max} \coloneqq \frac{\kernelcomputationexponent \cdot (\optalgexponent - \comptimeexponent)}{(\kernelcomputationexponent - 1) \cdot \compparameter}$ the upper bound for $\kernelsizeconstant$.
	Note that~$\alpha > \comptimeexponent \cdot \kernelcomputationexponent$ implies that
	\begin{align*}
      \lambda_{\max} \coloneqq \frac{\kernelsizeconstant_{\max} \compparameter}{ \kernelcomputationexponent} - \comptimeexponent & = \frac{\frac{\kernelcomputationexponent \cdot (\optalgexponent - \comptimeexponent)}{(\kernelcomputationexponent - 1) \cdot \compparameter}\compparameter}{ \kernelcomputationexponent} - \comptimeexponent
      = \frac{\optalgexponent - \comptimeexponent}{\kernelcomputationexponent - 1} - \comptimeexponent
      = \frac{\optalgexponent - \comptimeexponent - \comptimeexponent \cdot (\kernelcomputationexponent - 1)}{\kernelcomputationexponent - 1} = \frac{\optalgexponent - \comptimeexponent \cdot \kernelcomputationexponent}{\kernelcomputationexponent - 1} > 0
	\end{align*}
	Thus, we can pick~$\kernelsizeconstant < \kernelsizeconstant' < \kernelsizeconstant_{\max} $ such that $\lambda' := \frac{\kernelsizeconstant' \compparameter}{\kernelcomputationexponent} - \comptimeexponent> 0$.
\end{proof}

Note that if for $\mathcal{P}$ there is a (conditional) running time lower bound of~$\Omega(n^\alpha)$, then \Cref{thm:meta-composition} excludes (conditionally) running times of the form $\widetilde{O}(n_\mathcal{P} + \parameter^\kernelsizeconstant)$ for any $\kernelsizeconstant \in \mathbb{R}$ as $\lim_{\gamma \rightarrow 1} \frac{\kernelcomputationexponent \cdot (\optalgexponent - \comptimeexponent)}{(\kernelcomputationexponent - 1) \cdot \compparameter} = \infty$.
This running time also excludes linear-time computable polynomial-size kernels.
More precisely, we get the following.

\newcommand{\Qoptalgexponent}{\xi}
\begin{corollary}\label{cor:meta-kernel}
	Let~$\alpha > \comptimeexponent \ge 1$, $\kernelcomputationexponent > 1$, and~$\compparameter > 0$ with $\alpha > \comptimeexponent \cdot \kernelcomputationexponent$.
	Let~$\mathcal{P}$ be an $\alpha$-OR-decomposable problem with parameter~$n_\mathcal{P}$ that $(\comptimeexponent,\compparameter)$-OR-cross-composes into a problem~$\mathcal{Q}$ with parameters~$n_\mathcal{Q}$ and~$\parameter_\mathcal{Q}$.
	Assume that there is an~$\widetilde{O}( n_\mathcal{Q}^{\Qoptalgexponent})$~algorithm for deciding~$\mathcal{Q}$ and that~$n_\mathcal{Q}$ is upper bounded by the input size.
	If there exists an $\widetilde{O}(\parameter_\mathcal{Q}^\kernelsizeconstant)$-size $\widetilde{O} (n_\mathcal{Q}^{\kernelcomputationexponent})$-time kernel for~$\mathcal{Q}$ for some $\kernelcomputationexponent > 1$, $\kernelsizeconstant \in \mathbb{R}$, and
	\[0 < \kernelsizeconstant < \frac{\kernelcomputationexponent \cdot (\optalgexponent - \comptimeexponent)}{(\kernelcomputationexponent - 1) \cdot \compparameter \cdot \Qoptalgexponent},\]
	then~$\mathcal{P}$ can be solved in~$O(n_\mathcal{P}^{\alpha - \vareps})$ time for some~$\vareps > 0$.
\end{corollary}

\begin{proof}
	An $\widetilde{O}(\parameter_\mathcal{Q}^\kernelsizeconstant)$-size $\widetilde{O} (n_\mathcal{Q}^{\kernelcomputationexponent})$-time kernel together with an $\widetilde{O}( n_\mathcal{Q}^{\Qoptalgexponent})$-time algorithm solving $\mathcal{Q}$ yields an $\widetilde{O}(n_\mathcal{Q}^{\kernelcomputationexponent} + \parameter_\mathcal{Q}^{\kernelsizeconstant \Qoptalgexponent})$ algorithm for~$\mathcal{Q}$.
	The corollary now follows from \Cref{thm:meta-composition}.
\end{proof}

Our general approach to apply our framework is follows. 
Start with a problem~$\mathcal{P}$ that (under some hypothesis) cannot be solved in $O(n^{\alpha-\vareps})$ time for~$\vareps > 0$.
Then construct an $\alpha$-decomposition for~$\mathcal{P}$ followed by a (1,1)-OR-cross-composition into the target problem.

\section{OR-decomposable problems}
\label{sec:decomp}

In order to apply our framework, we first need some OR-decomposable problems.
We will observe that three fundamental problems from fine-grained complexity, namely \textsc{Orthogonal Vectors}, \textsc{3-SUM} and \textsc{Negative-Weight $k$-Clique}, are OR-decomposable.
These problems are also our source for running time lower bounds: 
Note that the former two problems cannot be solved in~$O(n^{2-\vareps})$ time unless the SETH respectively 3-Sum-hypothesis fail~\cite{VW18b}.
We moreover use that \textsc{Negative-Weight $3$-Clique} ($={}$\nt{}) cannot be solved in~$O(n^{3-\vareps})$ time unless APSP-hypothesis fails \cite{DBLP:journals/jacm/WilliamsW18}.

We will use that ``Many problems, like SAT, have a simple self-reduction proving that the “Direct-OR” version is hard, assuming the problem itself is hard''~\cite{DBLP:journals/talg/AbboudBHS22}.
This self-reduction is based on partitioning the instance into many small ones, with at least one of them containing the small desired structure (i.e., a pair of orthogonal vectors, three numbers summing to 0, or a negative triangle).

\subparagraph*{Orthogonal Vectors.}
We now show that \textsc{Orthogonal Vectors} is 2-OR-decomposable.

\begin{lemma}\label{lem:OV}
  \textsc{Orthogonal Vectors} parameterized by the number of vectors is 2-OR-decomposable.
\end{lemma}

\begin{proof}
	Let~$(I, n)$ be an instance of \textsc{Orthogonal Vectors} and~$\lambda > 0$.
	Set $\samplesizeexponent := \frac{2}{2 + \lambda}$.
	Partition $A$ into~$z:= \lceil n^{1- \samplesizeexponent} \rceil$ many sets~$A_1, \dots, A_{z}$ of at most $\lceil n^{\samplesizeexponent} \rceil$ vectors each.
	Symmetrically, partition $B$ into~$B_1, \dots, B_{z}$ of at most $\lceil n^{\samplesizeexponent} \rceil$ vectors each.
	We assume without loss of generality that $|A_i| = |B_j| = \lceil n^{\samplesizeexponent} \rceil =: n'$ (this can be achieved e.g.\ by adding the all-one vector).
	For each pair $(i, j) \in [z]^2$, create an instance~$(I_{(i,j)}, n') = ((A_i, B_j), n')$ of \textsc{Orthogonal Vectors}.
	We claim that this constitutes a 2-OR-decomposition.
	
	The number of vectors~$n'$ of each instance~$I_{(i, j)}$ is~$n' = O( n^{\nicefrac{2}{(2 + \lambda)}})$.
	The number of created instances is~$z^2 = {O} (n^{2\cdot (1- \samplesizeexponent)}) = {O} (n^{2 \cdot (1 - \nicefrac{2}{(2 + \lambda)})}) = {O} (n^{\nicefrac{(4 + 2\lambda - 4 )}{(2 + \lambda)}}) = {O} (n^{\nicefrac{2\lambda}{(2 + \lambda)}})$.
	The running time to compute the decomposition is $\widetilde O(z^2 \cdot n') = \widetilde O(n^{2- \epsilon})$.
	
	It remains to show that $(I, n)$ is a \yes-instance if and only if $(I_{(i,j)}, n')$ is a \yes-instance for some~$(i,j)  \in [z]^2$.
	First assume that $(I, n)$ is a \yes-instance.
	Then there exists some~$a \in A$ and~$b\in B$ such that $a $ and $b$ are orthogonal.
	Let $i \in [ z]$ such that $a \in A_i$ and $j \in [ z]$ such that $b \in B_j$.
	Then $a$ and $b$ are orthogonal vectors in~$(I_{(i, j)}, n')$, showing that $(I_{(i, j)}, n')$ is a \yes-instance.
	
	Finally, assume that there exists $(i^*, j^*) \in [ z]^2$ such that $(I_{(i^*, j^*)}, n')$ is a \yes-instance.
	Then there exists $a \in A_{i^*}$ and $b \in B_{j^*}$ which are orthogonal.
	Consequently, $a$ and $b$ are orthogonal vectors in~$I$, implying that $(I, n)$ is a \yes-instance.
\end{proof}

\begin{remark}
  Note that the above decomposition does not change the dimension~$d$. 
  Thus, even restricted versions of \textsc{Orthogonal Vectors} with~$d \in O(\log n)$ are 2-OR-decomposable.
  Further, we can assume that all constructed instances of \OV\ have the same number of vectors and the same dimension.
\end{remark}

\subparagraph*{Negative-Weight $k$-Clique.} \label{ssec:neg-wei-clique}

We now show that \textsc{Negative-Weight $k$-Clique} is $k$-OR-decomposable.

\begin{lemma}
	\label{lem:sampling}
	For any $k \ge 3$, \textsc{Negative-Weight $k$-Clique} parameterized by the number of vertices is $k$-OR-decomposable.
\end{lemma}

{
\begin{proof}
	The proof follows the ideas from \Cref{lem:OV}.
	Let~$(I = (G, w), n)$ be an instance of \textsc{Negative-Weight $k$-Clique} and~$\lambda > 0$.
	Set $\samplesizeexponent := \frac{k}{k + \lambda}$.
	Partition the set~$V(G)$ of vertices into~$z:= \lceil n^{1- \samplesizeexponent} \rceil$ many sets~$V_1, \dots, V_{z}$ of size at most $\lceil n^{\samplesizeexponent} \rceil$.
	We assume without loss of generality that $|V_i| = \lceil n^{\samplesizeexponent} \rceil$ for all $i\in [z]$ (this can be achieved e.g.\ by adding isolated vertices).
	Let $n_{(i_1, \dots, i_k)} := |\{i_1,\dots, i_k\}| \cdot |V_1| $.
	For each tuple $(i_1, \dots, i_k) \in [z]^k$, create an instance~$(I_{(i_1,\dots, i_k)} = G[A_{i_1} \cup \dots \cup A_{i_k}], n_{(i_1, \dots, i_k)})$ of \textsc{Negative-Weight $k$-Clique}.
	We claim that this constitutes a $k$-OR-decomposition.
	
	Each instance~$(I_{(i, j)}, n_{(i_1, \dots, i_k)})$ has at most~${O} (k \cdot n^{\samplesizeexponent}) = {O} (n^{\nicefrac{k}{(k + \lambda)}})$ vertices (note that $k$ is a constant here).
	The number of created instances is~$z^k = {O} (n^{k\cdot (1- \samplesizeexponent)}) = {O} (n^{k \cdot (1 - \nicefrac{k}{(k + \lambda)})}) = {O} (n^{\nicefrac{(k^2 + k\lambda - k^2)}{(k + \lambda)}}) = {O} (n^{\nicefrac{k\lambda}{(k + \lambda)}})$.
	Further, the summed size of all created instances and therefore also the running time is $\widetilde O(z^k \cdot k \cdot n^{2 \samplesizeexponent}) =\widetilde  O( n^{k - k \samplesizeexponent + 2 \samplesizeexponent}) = \widetilde O(n^{k - (k-2)\samplesizeexponent}) = \widetilde O (n^{k'} ) $ for some~$k'< k$.
	
	It remains to show that $(I, n)$ is a \yes-instance if and only if $(I_{(i_1,\dots, i_k)}, n_{(i_1, \dots, i_k)})$ is a \yes-instance for some~$(i_1, \dots,i_k)  \in [ z]^k$.
	First assume that $(I, n)$ is a \yes-instance.
	Thus, $G$ contains a negative-weight clique~$C = \{v_1, \dots, v_k\}$.
	Let $i_j$ such that $v_j \in V_{i_j}$.
	Then $C$ is a negative-weight $k$-clique in~$(I_{(i_1, \dots, i_k)}, n_{(i_1, \dots, i_k)})$.
	
	Finally, assume that there exists $(i_1, \dots, i_k) \in [ n]^k$ such that $(I_{(i_1, \dots, i_k)}, n_{(i_1, \dots, i_k)})$ is a \yes-instance.
	Then there exists a clique in $G[V_{i_1} \cup \dots \cup V_{i_k}]$.
	Since $G$ contains $G[V_{i_1} \cup \dots \cup V_{i_k}]$, also $G$ contains a negative-weight $k$-clique.
\end{proof}
}

\subparagraph{3-Sum.}
Showing that \textsc{3-Sum} is 2-OR-decomposable requires some more work. 

\begin{lemma}
	\label{lem:3sum-decomp}
	\textsc{3-Sum} parameterized by the number of numbers is 2-OR-decomposable.
\end{lemma}

{
Instead of directly considering \textsc{3-Sum}, we consider a variation called \textsc{Convolution 3-Sum} of \textsc{3-Sum} which was shown to be equivalent (under quasi-linear time Las-Vegas reductions) to \textsc{3-Sum}~\cite{DBLP:conf/soda/KopelowitzPP16}.

\decprob{\textsc{Convolution 3-Sum}}
{An array~$A$ of $n$ integers.}
{Are there $i, j\in [n]$ such that $A[i] + A[j] = A[i+j]$?}
For simplicity, we will not directly work with \textsc{Convolution 3-Sum}, but with a ``multicolored'' version of it:
\decprob{\mcconv}
{Three arrays~$A$, $B$, and $C$, each of $n$ integers.}
{Are there $i, j\in [n]$ such that $A[i] + B[j] = C[i+j]$?}

First, we show how to reduce \textsc{Convolution 3-Sum} to \mcconv\ and vice versa.

\begin{lemma}\label{lem:conv-mc}
 There is a linear-time parameter-preserving reduction from \textsc{Convolution 3-Sum} parameterized by the number of numbers to \mcconv\ parameterized by the number of numbers and vice versa.
\end{lemma}

\begin{proof}
 Let~$(A, n)$ be an instance of \textsc{Convolution 3-Sum}.
 We create an instance~$((A, B, C), n)$ of \mcconv\ by setting $B[i] := A[i]$ and $C[i] := A[i]$.
 Note that for~$i, j \in [n]$, we have $A[i] + A[j] = A[i + j]$ if and only if $A[i] + B[i] = A[i] + A[j] = C[i + j]$, implying that the two instances are equivalent.
 
 Now consider an instance~$(I = (A, B, C), n)$ of \mcconv.
 Let $z_{\max} := \max_{i\in [n]} \max \{A[i], B[i], C[i]\} + 1$.
 We create an instance~$(A', n'=4n)$ of \textsc{Convolution 3-Sum} as follows.
 The array~$A'$ has length $4n$ with
 \[
  A' [i] \coloneqq \begin{cases}
            - 5 z_{\max} & 1 \le i \le n\\
            A[i - n] + z_{\max} & n + 1 \le i \le 2n,\\
            B[i-2n ]+ 3z_{\max} & 2n+1 \le i \le 3n,\\
            C[i- 3n] + 4z_{\max} & 3n + 1 \le i \le 4n\,.
           \end{cases}
 \]
 First, note that a solution $A[i] + B[j] = C[i + j]$ to~$(I, n)$ implies a solution to~$(A', n')$ as $A' [i+ n] + A' [j+ 2n] = A[i] +z_{\max} + B[j] + 3z_{\max}= C[i+j] + 4 z_{\max} = A' [ 3n + i + j]$.
 Second, consider a solution $A'[i'] + A'[j'] = A'[i' + j']$.
 Then $i', j' > n$ as otherwise $A'[i'] + A'[j'] < 0 $ and $A'[i'] + A'[j'] \neq - 5 z_{\max}$.
 Further, $i', j' \le 3n$ as otherwise $A'[i'] + A'[j'] > 5 z_{\max} > A'[i' + j']$.
 Further, we have without loss of generality $n+1 \le i' \le 2n$ and $2n+1 \le j' \le 3n$.
 Thus, we have $A[i' - n] + B[j' - 2n] = A' [ i'] - z_{\max} + A'[j' ] - 3 z_{\max} = A[i' + j'] - 4 z_{\max} = C[i' +j' - 3n]$. 
\end{proof}

Next, we show that \mcconv\ is 2-OR-decomposable.

\begin{lemma}\label{lem:3sum-convolution}
  \mcconv\ parameterized by the number of numbers is 2-OR-decomposable.
\end{lemma}

\begin{proof}
	The idea behind the proof is essentially the same as for \Cref{lem:OV}.
	For an array~$A$ we denote the array consisting of the entries $(A[i], A[i+1], \dots, A[j])$ by $A[i,j]$.
	Let $(I = (A, B, C), n)$ be an instance of \mcconv\ and $\lambda > 0$.
	Set $\samplesizeexponent := \frac{2}{2 + \lambda}$ and $q := \lceil n^\samplesizeexponent \rceil$.
	Partition $A$ into~$z:= \lceil n^{1- \samplesizeexponent} \rceil$ many arrays~$A_1 = A[1,q], A_2 = A[q+1, 2q], \dots, A_{z} = A[(z-1)q + 1, n]$.
	For simplicity, we assume that $n = z \cdot q$.
	Symmetrically, partition $B$ into $z$ many arrays~$B_1 = B[1,q], B_2 = B[q+1, 2q], \dots, B_{z} = B[(z-1)q + 1, n]$.
	For~$(\ell, p) \in [z]^2$ we denote the interval~$C[(\ell + p - 2) q + 2, (\ell + p)q]$ by~$C_{(\ell, p)}$.
	For each pair $(\ell, p) \in [z]^2$, create an instance~$(I_{(\ell,p)} = (A_\ell, B_p, C_{(\ell, p)}), q)$ of \textsc{Convolution 3-Sum}.
	We claim that this constitutes a 2-OR-decomposition.
	
	Each instance~$(I_{(\ell, p)}, q)$ has at most~${O} (q) =  O (n^{\samplesizeexponent}) = {O} (n^{\nicefrac{2}{(2 + \lambda)}})$.
	The number of created instances is~$z^2 = {O} (n^{2\cdot (1- \samplesizeexponent)}) = {O} (n^{2 \cdot (1 - \nicefrac{2}{(2 + \lambda)})}) = {O} (n^{\nicefrac{(4 + 2\lambda - 4)}{(2 + \lambda)}}) = {O} (n^{\nicefrac{2\lambda}{(2 + \lambda)}})$.
	Computing the decomposition can be done in $\widetilde O (z^2 \cdot q)= \widetilde O(n^{2-\epsilon})$ time.
	
	It remains to show that $(I, n)$ is a \yes-instance if and only if $(I_{(\ell,p)}, q)$ is a \yes-instance for some~$(\ell, p)  \in [ z]^2$.
	First assume that $(I, n)$ is a \yes-instance.
	Then there exists some~$i = (\ell - 1) q + i^*$, $j = (p -1) q + j^*$ such that $A[i] + B[j] = C [ i + j]$.
	Note that $A[i] = A_\ell [i^*]$, $B[j ] = B_p[j^*]$, and $C[ i + j] = C_{(\ell, p)}[ i^* + j^*]$.
	Thus, $(I_{(\ell, p)}, q)$ is a \yes-instance.
	
	Finally, assume that there exists $(\ell, p) \in [ z]^2$ such that $(I_{(\ell, p)}, q)$ is a \yes-instance.
	Then $A_\ell [i^* ] + B_p[j^*] = C_{(\ell, p)} [ i^* + j^*]$ for some~$i^*, j^* \in [q]$.
	Consequently, $A[(\ell -1) q + i^*] + B[ (p - 1) q + j^*] = C_{(\ell, p)} [i^* + j^*] = C[(\ell + p - 2)q + i^* + j^*]$, implying that $(I, n)$ is a \yes-instance.
\end{proof}

The 2-OR-decomposability of \textsc{3-Sum} now follows from the quasi-linear time parameter-preserving reductions between \textsc{3-Sum} and \textsc{Convolution 3-Sum}~\cite{DBLP:conf/soda/KopelowitzPP16} as well as between \textsc{Convolution 3-Sum} and \mcconv~\Cref{lem:conv-mc}.

{
\begin{proof}[Proof of \cref{lem:3sum-decomp}]
 Follows from \Cref{lem:3sum-convolution}, the equivalence of \textsc{3-Sum} and \textsc{Convolution 3-Sum}~\cite{DBLP:conf/soda/KopelowitzPP16}, the equivalence of \textsc{Convolution 3-Sum} and \mcconv, and \Cref{prop:decompose}.
\end{proof}
}
}

\subparagraph{Applying the framework.}
The above results make our framework easier to apply.
To apply \cref{thm:meta-composition} we only need to provide a suitable OR-cross composition from one of the three problems discussed above.
We thus arrive at the following.

\begin{proposition}\label{prop:apply-framework}
	Let~$\mathcal{Q}$ be a problem with parameters~$n_\mathcal{Q}$ and~$\ell_\mathcal{Q}$.
	If \textsc{Orthogonal Vectors} resp.\ \textsc{3-Sum} parameterized by~$n$ $(1,1)$-OR-cross-composes into~$Q$, then an~$O(\ell_\mathcal{Q}^{\beta} + n_\mathcal{Q}^\gamma)$-time algorithm for~$\mathcal{Q}$ for any~$2 > \gamma > 1$ and~$\beta < \nicefrac{\gamma}{ \gamma - 1}$ refutes the SETH respectively the 3-Sum-Hypothesis.

	If \textsc{Negative Triangle} parameterized by the number of vertices $(1,1)$-OR-cross-composes into~$Q$, then an~$O(\ell_\mathcal{Q}^{\beta} + n_\mathcal{Q}^\gamma)$-time algorithm for~$\mathcal{Q}$ for any~$3 > \gamma > 1$ and~$\beta < \nicefrac{2\gamma}{ \gamma - 1}$ refutes the APSP-Hypothesis.
\end{proposition}

\section{Applications} 
\label{sec:applications}

We now apply our framework from \Cref{thm:meta-composition} to several problems from different areas such as string problems (\Cref{sec:strings}), computational geometry~(\Cref{sec:cg}), and subgraph isomorphism (\Cref{sec:iso}).

\subsection{String problems}
\label{sec:strings}
We start with \LCS{} that can be solved in $O(n^2)$ time algorithm~\cite{CLRS09}.
Assuming SETH, there is no algorithm solving \LCS{} in~$O(n^{2-\vareps})$ time for any~$\vareps > 0$~\cite{BK15,ABW15}.
However, \LCS\ can be solved in $O(kn + n \log n)$ time, where $k$ is the length of the longest common subsequence~\cite{Hir77}.
\citet{BK18} proved that this running time is essentially optimal under SETH.
Here we reprove this fact using our framework.
We refer to \citet{BK18} for a more extensive literature review on \LCS.

A \emph{string} over an alphabet~$\Sigma$ is an element from~$\Sigma^*$.
We access the $i$-th element of a string~$x$ via $x[i]$.
A \emph{subsequence} of a string~$x$ is a string~$y$ such that there is an injective, strictly increasing function~$f$ with $y[i] = x[f(i)]$ for all~$i$.
A \emph{common subsequence} of two strings~$x$ and $x'$ is a string which is a subsequence of both~$x$ and~$x'$.
For two strings~$x$ and $y$, we denote their concatenation (i.e.\ the string starting with~$x$ and ending with~$y$)  by~$x \circ y$.

\begin{lemma}\label{lem:lcs-comp}
 \textsc{Orthogonal Vectors} with logarithmic dimension parameterized by the number of vectors $(1,1)$-OR-cross-composes into \LCS\ parameterized by the length of the input strings and the length~$k$ of the longest common subsequence.
\end{lemma}

\begin{proof}
 Let~$(I_1^{\text{OV}}, n_1), \dots, (I_t^{\text{OV}}, n_t)$ be instances of \textsc{Orthogonal Vectors} with logarithmic dimension.
 We denote
 by $d_i$ the dimension of~$I_i^{\text{OV}}$.
 Abboud, Backurs, and Williams~\cite{ABW15} gave a reduction from \textsc{Orthogonal Vectors} to \LCS\ which, given an instance of \textsc{Orthogonal Vectors} with $n$ vectors of dimension~$d$, constructs in~$n \cdot d^{O(1)}$ time an equivalent instance of \LCS\ with strings of length~$n \cdot d^{O(1)}$.
 We apply this reduction to each instance~$(I_i^{\text{OV}}, n_i)$ of \textsc{OV}, giving us an instance~$I_i^{\text{LCS}} = ((x_i^1, x_i^2), (n_i^{\text{LCS}}, k_i))$ with $k_i = O(n_i \cdot d_i^{O(1)})$ and $n_i^{\text{LCS}} = n_i \cdot d_i^{O(1)}$.
 We assume that $k_i = k_j$ for every $i, j$ (this can be achieved by appending identical sequences of appropriate length to strings~$x_i^1$ and $x_i^2$), and set $k:= k_i$.
 Further, we assume that the alphabets used for $I_i^{\text{LCS}}$ and~$I_j^{\text{LCS}}$ are disjoint for $i \neq j$.
 We define $x^1 \coloneqq x_1^1 \circ x_2^1 \circ \dots \circ x_t^1$ and~$x^2 \coloneqq x_t^2 \circ x_{t-1}^2 \circ \dots \circ x_1^2$.
 The OR-cross-composition constructs the instance $(x^1, x^2, k)$ (the parameter is $(n^{\text{LCS}}, k)$ with $n^{\text{LCS}} = \sum_{i=1}^t n_i^{\text{LCS}}$).
 
 We now show correctness of the reduction.
 First assume that $(I_i^{\text{OV}}, n_i)$ is a \yes-instance for some~$i\in [t]$.
 Then $x_i^1$ and $x_i^2$ contain a subsequence~$y$ of length $k_i = k$.
 This subsequence~$y$ is also a subsequence of $x^1$ and $x^2$, so $(x^1, x^2, k)$ is a \yes-instance.
 Vice versa, assume that $(x^1, x^2, k)$ is a \yes-instance, i.e., $x^1$ and $x^2$ contain a subsequence~$y$ of length~$k$.
 Let~$i \in [t]$ such that the first letter of~$y$ is contained in~$x^1_i$.
 We claim that all letters from~$y$ are contained in~$x^1_i$.
 Since the first letter~$y[1]$ of~$y$ is contained in~$x^1_i$, no letter of~$y$ can be contained in~$x^1_j$ for $j<i$.
 For~$j> i$, note that (as any letter from~$x_j^2$ only appears in $x_j^2$ and $x_j^1$) any letter from~$x_j^1$ appears only before~$x^2_i$ in $x_2$ and thus would have to appear before~$y[1]$ in $y$, a contradiction to~$y[1]$ being the first letter of~$y$.
 Thus, $y$ is contained in $x^1_i$ and $x_i^2$.
 Consequently, $(x^1_i, x_i^2, k)$ is a \yes-instance, implying that $(I_i, n_i)$ is a \yes-instance.
 
 Computing the OR-cross-composition can be done in $t \cdot \max_{i\in t} n_i \cdot (\max_{i \in [t]} d_i)^{O(1)} = \widetilde O (t \max_{i \in [t]} n_i)$ time and $k = O(\max_{i\in t} n_i \cdot (\max_{i\in t} d_i)^{O(1)}) = \widetilde O (\max_{i=1}n_i )$ (we use here that $d_i \in O (\log n_i)$).
 Further, we have $n^{\text{LCS} } = \sum_{i=1}^t n_i ^{\text{LCS}} = O (\sum_{i=1}^t n_i \cdot d^{O(1)}_i ) = \widetilde O (\sum_{i=1}^t)$.
 Thus, it is an $(1, 1)$-OR-cross-composition.
\end{proof}

We remark that the use of alphabets of non-constant size in the above composition is probably necessary as it excludes a linear-time computable kernel of polynomial size.
In contrast, \LCS{} \emph{with} constant alphabet size parameterized by solution size admits a linear-time computable polynomial-size kernel (note that the SETH-based $O(n^{2-\epsilon})$ lower bound also holds for binary alphabets~\cite{BK18}):

\begin{proposition}
	\label{prop:lcs-constant-alphabet-kernel}
	\LCS\ with constant alphabet size parameterized by solution size~$k$ admits a linear-time computable polynomial-size kernel.
\end{proposition}

{
\begin{proof}
  Consider the following reduction rule which directly leads a polynomial kernel.
  A \emph{substring} of a string~$x$ is a string~$y$ such that $y = (x[i], x[i+1], \dots, x[j])$ for some~$i < j$.
 
 \begin{rrule}\label{rrule:frequent}
	If, for some~$t \in \N$ and $i \in \{1,2\}$, a string $x^i$ contains a substring~$x'$ of length at least~$(k+ 1)^{t}$ over only $t$ different letters and not containing a substring~$x''$ of length~$(k+1)^{t-1}$ over at most~$t-1$ different letters, then we can delete all but the first~$(k + 1)^{t}$ letters of~$x'$.
 \end{rrule}

  First, we argue that the reduction rule is safe.
  Let $I = (x^1, x^2, k)$ be an instance of \LCS.
  We assume without loss of generality that the reduction rule can be applied to~$x^1$, i.e., $x^1$ contains a substring~$x'$ of length~$(k+1)^t$ which contains only~$t$ different letters, and all substrings of~$x'$ containing only $t' < t$ different letters have length at most~$(k+1)^{t'}$.
  Let~$I_{\reduced} = (x^1_{\reduced}, x^2, k)$ be the instance arising from applying \Cref{rrule:frequent}.
  Clearly, if $I$ is a \no-instance, then also $I_{\reduced}$ is also a \no-instance.
  
  Now assume that $I$ is a \yes-instance, i.e., $x^1 $ and $x^2$ contain a common subsequence~$z$ of length $k$.
  By assumption $x'$ contains no substring of length~$(k + 1)^{t-1}$ which contains at most~$t-1$ different letters.
  Thus, in the first $(k + 1)^{t-1} -1$ letters of~$x'$, each of the $t$ letters appears at least once.
  More generally, for any~$i \in [k]$, each letter appears at least once among $x'[(i-1) \cdot (k + 1)^{t-1} + 1], \dots, x'[i\cdot (k + 1)^{t-1}]$.
  Consequently, each string of length~$k$ over the~$t$ appearing letters is a subsequence of~$x'$.
  Hence, $z$ is also a subsequence of~$x^1_{\reduced}$ (and thus a common subsequence of $x^1_{\reduced}$ and $x^2$), showing that $I_{\reduced}$ is also a \yes-instance.
  Next, we analyze the time needed to exhaustively apply \cref{rrule:frequent}.
 
 \begin{claim}
  For constant alphabet size, \cref{rrule:frequent} can be exhaustively applied in linear time.
 \end{claim}

 \begin{claimproof}
  We process~$x^1$ and~$x^2$ from left to right.
  For each subset~$S$ of the alphabet, we have a variable storing the number of consecutive letters from~$S$ are before the current letter.
  If this number exceeds $(k+1)^{|S|}$ for some letter, then we delete this letter.
  Note that whenever we found a string of length~$(k+1)^{|S|}$ containing only letters from~$S$, then this string does not contain a substring of length~$(k+1)^{|S|-1}$ containing only letters from~$S\setminus \{\ell\}$ for some~$\ell \in S$ as in this case, the number stored for~$S\setminus \{\ell\}$ would have exceeded~$(k+1)^{|S|-1}$.
  
  The above procedure clearly runs in linear time as the alphabet size is constant.
 \end{claimproof}
 
 After the exhaustive application of \cref{rrule:frequent}, the size of the instance is clearly $\mathcal{O} (k^{|\Sigma|})$.
 As $|\Sigma|$ is constant, this is a polynomial-sized kernel.
\end{proof}
}

Analogously to \LCS, we derive similar hardness results for the related problems \textsc{Longest Common Weakly Increasing Subsequence} and \textsc{Longest Common Increasing Subsequence}. 
Both problems can be solved in slightly subquadratic time~\cite{DBLP:conf/stacs/Duraj20,DBLP:conf/isaac/AgrawalG20}.
For \textsc{Longest Common Increasing Subsequence}, our lower bound was already shown by \citet{DBLP:journals/algorithmica/DurajKP19}.
We refer to \citet{DBLP:journals/algorithmica/DurajKP19,DBLP:conf/stacs/Duraj20} for a broader literature review on \textsc{Longest Common Subsequence}.

\begin{lemma}
	\label{lem:increasing-lcs-comp}
 \textsc{Orthogonal Vectors} with logarithmic dimension $(1,1)$-OR-cross-composes into \textsc{Longest Common (Weakly) Increasing Subsequence} parameterized by the length~$n$ of the strings and the length~$k$ of the longest common (weakly) increasing subsequence.
\end{lemma}

{
\begin{proof}
 The proof is analogous to the proof of \Cref{lem:lcs-comp}.
 Let~$(I_1^{\text{OV}}, n_1), \dots, (I_t^{\text{OV}}, n_t)$ be instances of \textsc{Orthogonal Vector} with logarithmic dimension.
 We denote
 by $d_i$ the dimension of~$I_i^{\text{OV}}$.
 Polak~\cite{DBLP:journals/ipl/Polak18} and Duraj, K\"unnemann, and Polak~\cite{DBLP:journals/algorithmica/DurajKP19} gave a reduction from \textsc{Orthogonal Vectors} to \textsc{Longest Common (Weakly) Increasing Subsequence} which, given an instance of \textsc{Orthogonal Vectors} with $n$ vectors of dimension~$d$, constructs in~$\widetilde O (n \cdot d^{O(1)})$ time an instance of \textsc{Longest Common (Weakly) Increasing Subsequence} of length~$\widetilde O (n \cdot d^{O(1)})$.
 We apply this reduction to each instance~$(I_i^{\text{OV}}, n_i)$ of \textsc{OV}, giving us an instance~$I_i^{\text{LCIS}} = (x_i^1, x_i^2, k_i)$ with $k_i = n_i \cdot d_i^{O(1)}$ and whose strings have length $\widetilde O (n \cdot d^{O(1)})$.
 We assume that $k_i = k_j$ for every $i, j$ (this can be achieved by appending identical sequences of increasing numbers of appropriate length to strings~$x_i^1$ and $x_i^2$), and set $k:= k_i$.
 Let~$C_{\max} $ be the largest number appearing in any of the instances~$I^{\text{LCIS}}_1$, \dots, $I^{\text{LCIS}}_t$.
 We modify instance~$I_i^{\text{LCIS}}$ by increasing each number by~$i \cdot C_{\max}$.
 We define $x^1 \coloneqq x_1^1 \circ x_2^1 \circ \dots \circ x_t^1$ and~$x^2 \coloneqq x_t^2 \circ x_{t-1}^2 \circ \dots \circ x_1^2$.
 The OR-cross-composition constructs the instance $(x^1, x^2, k)$.
 
 The correctness of the reduction can be proven in analogy to \Cref{lem:lcs-comp}.
 
 The OR-cross-composition can be computed in $\widetilde O (t \cdot \max_{i\in t} n_i \cdot (\max_{i \in [t]} d_i)^{O(1)})  = \widetilde O (t \max_{i \in [t]} n_i)$ time and $k = \max_{i\in t} n_i \cdot (\max_{i\in t} d_i)^{O(1)} = \widetilde O (\max_{i=1}n_i )$ (we use here that $d_i \in O (\log n_i)$).
 Thus, it is an $(1, 1)$-OR-cross-composition.
\end{proof}
}

Combining \Cref{prop:apply-framework} and \cref{lem:lcs-comp,lem:increasing-lcs-comp}, we get the following result:

\begin{corollary}\label{cor:LCSlower}
 Unless the SETH fails, there is no $\widetilde O(n^\gamma + k^\beta)$-time algorithm for \LCS\ or \textsc{Longest Common (Weakly) Increasing Subsequence} parameterized by solution size~$k$ for $1 < \gamma < 2$ and $ \beta < \frac{\gamma}{ \gamma - 1}$.
 
 Unless the SETH fails, there is no $\widetilde O(n^\gamma)$-time, $\widetilde O(k^\beta)$-size kernel for \LCS\ or \textsc{Longest Common (Weakly) Increasing Subsequence} parameterized by solution size for $1 < \gamma < 2$ and $ \beta < \frac{\gamma}{2 \cdot (\gamma - 1)}$.
\end{corollary}

We remark that the running time lower bounds are tight;
the tight upper bound follows from the known $O(k n + n \log n)$ time algorithm for \LCS~\cite{Hir77} or the $O(n k \log\log  n + n \log n)$ time algorithm for \textsc{Longest Common (Weakly) Increasing Subsequence}~\cite{DBLP:journals/jda/KutzBKK11} and \Cref{thm:upper-bound}.

\subsection{Computational Geometry}
\label{sec:cg}

We now turn to problems from computational geometry.
We will denote the Euclidean distance between two points~$p$ and~$q$ in the plane by~$\dist (p, q)$.
We start with some definitions related to the discrete Fr\'echet distance.

For a list of points~$P$, a \emph{tour} through~$P$ is a surjective, non-increasing function~$f_P : [2n]\rightarrow [n]$ such that $f_P (1) = 1$ and $f_P (2n) = n$.
The input of \textsc{Discrete Fr\'echet Distance} contains two lists of points~$P = (p_1, \dots, p_n)$ and~$Q = (q_1, \dots, q_n)$.
A pair~$(f_P, f_Q)$ of tours through~$P$ and $Q$ is called a \emph{traversal}.
We call $\max_{i \in [2n]}\dist (p_{f_P (i)}, q_{f_Q (i)})$ the \emph{width} of the traversal.
We say that a traversal~$(f_P, f_Q)$ \emph{traverses} a sublist~$p_i, p_{i+1}, \dots, p_j$ of~$P$ (or $q_i, q_{i+1}, \dots, q_j$ of~$Q$) if, for some~$\ell \in [2n] $ and $i^* \in [n]$), we have $f_P (\ell) = p_i$ and $f_Q (\ell) = i^*$, and for any~$\ell' \in [j-i]$, we have $f_P (\ell + \ell' ) = {i+\ell'}$ and $f_Q (\ell + \ell') = i^*$ (or $f_P (\ell) = i^*$ and $f_Q (\ell) = i$, and for any $\ell' \in [j-i]$, we have $f_P (\ell + \ell' ) = i^*$ and $f_Q (\ell + \ell' ) = {i +\ell'}$).

\textsc{Discrete Fr\'echet Distance} can be solved in $O(n^2\cdot \log\log n / \log^2n )$ time~\cite{DBLP:journals/siamcomp/AgarwalAKS14}.
\citet{DBLP:conf/focs/Bringmann14} gave a linear-time reduction from \textsc{Orthogonal Vectors}\footnote{Bringmann~\cite{DBLP:conf/focs/Bringmann14} reduces from \textsc{Satisfiability}, but his reduction implicitly reduces \textsc{Satisfiability} to \OV\ and then \OV\ to \textsc{Discrete Fr\'echet Distance}.} to \textsc{Discrete Fr\'echet Distance}, showing that assuming SETH, \textsc{Discrete Fr\'echet Distance} cannot be solved in $O(n^{2-\epsilon})$ time for any~$\epsilon > 0$.
We use this reduction to get a $(1,1)$-OR-composition for the parameter $\max_{i\in [2n]} |f_P (i) - f_Q (i)|$, i.e., the maximum number of time steps which $P$ may be ahead of~$Q$ (or vice versa) in the optimal solution.
We will call the parameter~$\max_{i\in [2n]} |f_P (i) - f_Q (i)|$ the \emph{maximum shift} of the traversal~$(f_P,f_Q)$.
For a more extensive literature review on \textsc{Discrete Fr\'echet Distance}, we refer to~\citet{DBLP:journals/siamcomp/AgarwalAKS14}.

\begin{lemma}
	\label{lem:ov-into-dft}
	\textsc{Orthogonal Vectors} admits a $(1, 1)$-OR-cross-composition into \textsc{Discrete Fr\'echet Distance} parameterized by the length of the input lists of points and the maximum shift in a solution.
\end{lemma}

{
Our OR-cross-composition is based on the reduction from \OV\ to \textsc{Discrete Fr\'echet Distance} by Bringmann~\cite{DBLP:conf/focs/Bringmann14} which has the following properties.

\begin{proposition}[{Bringmann~\cite[Section III.A]{DBLP:conf/focs/Bringmann14}}]\label{prop:bringmann}
  Given an instance~$((A, B), n)$ of \OV\, one can compute in $\widetilde O(n)$ time an instance $(P_A, P_B, 1)$ of \textsc{Discrete Fr\'echet Distance} such that
  \begin{enumerate}
   \item the Fr\'echet distance of~$P_A$ and~$ P_B$ is 1 if and only if $((A, B), n)$ is a \yes-instance,\label{cond:equiv}
   \item $P_A$ starts at $s_A \coloneqq (-\frac{1}{3}, \frac{1}{5})$ and $s_A$ appears only at the start of $P_A$,\label{cond:sA}
   \item $P_B$ starts at $s_B \coloneqq (-\frac{1}{3}, 0)$,\label{cond:sB}
   \item each point from~$P_B$ has distance at most 1 from~$c_A \coloneqq (0, \frac{1}{3})$,
   \label{cond:cA}
   \item each point from~$P_A$ has distance at most 1 from~$c_B \coloneqq s_B$,
   \label{cond:cB}
   \item $\ell_A \coloneqq (-\frac{4}{3}, 0)$ has distance larger than~1 to all points from~$Q\setminus \{s_B\}$, and
   \label{cond:lA}
   \item $\ell_B \coloneqq (-\frac{4}{3}, \frac{1}{5})$ has distance larger than~1 to all points of~$P_A\setminus \{s_A\}$.
   \label{cond:lB}
   \item 
   one point in~$P_B$ has distance more than 1 from~$s_A$.
\label{cond:no}
  \end{enumerate}
\end{proposition}

\begin{proof}[Proof of \cref{lem:ov-into-dft}]
Let $(I_1^{\text{OV}} = (A_1, B_1), n_1), \dots, (I_t^{\text{OV}}= (A_t , B_t), n_t)$ be instances of \textsc{Orthogonal Vectors}.
To simplify notation, we assume that $I_1^{\text{OV}}$, \dots, $I_t^{\text{OV}}$ all are of the same dimension~$d$ and that there is some~$n \in \N$ such that~$n_i = n$ for all~$i \in [t]$.
Further, we assume that no two vectors of different instances of~$I_t^{\text{OV}}$ are orthogonal.
This can be achieved by increasing the dimension by $2\log t$ and using the additional dimensions to add to each~$a \in A_i$ the binary encoding of~$i$ and its ``inverse'' (i.e., replacing 0's by 1's and vice versa).
To each $b \in B_i$, we append the ``inverse'' of the binary encoding of~$i$ and the binary encoding of~$i$.

We apply the reduction from \cref{prop:bringmann} to each instance~$(I_i^{\text{OV}}, n)$, resulting in an instance~$I_i^{\text{DFD}} = (P_A^i, P_B^i, 1)$ of \textsc{Discrete Fr\'echet Distance}.
We set $n_{i}^{\text{DFD}} := |P_A^i|$ to be the length of~$P_A^i$.
For~$i \in [t]$, let $s_A^i \coloneqq s_A, s_B^i \coloneqq s_B, c_A^i \coloneqq c_A, \ell_A^i \coloneqq \ell_A, \ell_B^i \coloneqq \ell_B$ (where $s_A, s_B, c_A, c_B, \ell_A$, and $\ell_B$ are be defined as in \Cref{prop:bringmann}).
Further, let $I_{t+1}^{\text{DFD}}$ be an instance arising from applying \Cref{prop:bringmann} to an arbitrary \no-instance of \OV.
From this, we construct an instance~$(I^{\text{DFD}} = (P_A, P_B, 1), (2t + \sum_{i=1}^{t+1} n_i^{\text{DFD}} , 2 \max_{i\in [t+1]} n_i^{\text{DFD}}))$ of \textsc{Discrete Fr\'echet Distance} parameterized by the length of the point lists and the maximum shift as follows.
We set $P_A \coloneqq \bigcirc_{i \in [t]} (P_A^i\circ \ell_A^i \circ c_A^i) \circ P_A^{t+1}$ and $P_B \coloneqq \bigcirc_{i \in [t+ 1]} (P_B^i \circ \ell_B^i) \circ c_B$.
To ensure that~$P_A$ and $P_B$ have the same number of points, we append an appropriate number of copies of~$c_B$ at the end of~$P_B$.

On an intuitive level, the idea of the construction is as follows: 
The asymmetry in the construction allows to traverse an arbitrary part of~$P_A$ while being in~$s_B^i = c_B$ in~$P_B$ for some~$i \in [t+1]$. 
However, the converse is not true and, thus, traversing in~$P_B$ as far as possible is desirable.
In order to traverse~$\ell_B^i$ we have to be in~$s_A^j$ or~$\ell_A^{j}$ for some~$i \in [t+1]$. 
We cannot stay in~$s_A^j$ since by \Cref{cond:no} at least one point in~$P_{B_i}$ has distance larger than one from~$s_A^i$.
Hence, there are two options: 
The traversal contains~$(s_A^j,\ell_B^i)$ for $i,j \in [t+1]$.
As \Cref{cond:no} forces us out of the~$s_A$-points, we will have $j>i$ when following this option.
Thus, only using this option results in getting stuck at the end as the last~$\ell_{B_{t+1}}$ cannot be traversed.
To prevent this and ``get ahead'' in~$P_B$, we have to use the second option at some point:
The traversal contains~$(\ell_A^i,\ell_B^i)$ for some~$i \in [t]$.
We prove that in this case~$P_{A_i}$ and~$P_{B_i}$ have to be traversed simultaneously, that is, $I_i^{\text{DFD}}$ is a yes-instance.
The benefit of the second option is that we can follow~$(\ell_A^i,\ell_B^i)$ by $(c_A^i,s_B^{i+1})$ and then traverse~$P_{B_{i+1}}$. 
Then, being ``ahead'' in~$P_B$, we can use the first option for the rest of the traversal. 

We now show the correctness of the OR-cross-composition.
\begin{claim}\label{claim:hin}
	If there is some~$i^* \in [t]$ such that there are $a \in A_{i^*}$ and $b \in B_{i^*}$ which are orthogonal, then there is a traversal~$T$ of $(P_A, P_B)$ of width at most 1 and maximum shift~$2 \max_{i\in [t+1]} n_i^{\text{DFD}}$.
\end{claim}

\begin{claimproof}
	We define a traversal as follows:
	We start with~$(s_A^1, s_B^1)$.
	For each $i < i^*$, we define the following traversal~$T_i$ (which first traverses~$P_A^i \circ \ell_A^i \circ c_A^i$ and afterwards $P_B^i \circ \ell_B^i$):
	Traverse~$P_A$ until the end of~$P_A^i$.
	By \Cref{cond:cB}, this traversal so far has width~1.
	Afterwards, go to~$(\ell_A^i, s_B^i)$, followed by~$(c_A^i, s_B^i)$.
	Then traverse~$P_B^i$.
	By \Cref{cond:cA}, this traversal still has width~1.
	We continue with~$(s_A^{i+1}, \ell_B^i)$ and~$(s_A^{i+1}, s_B^{i+1})$.
  
  By \Cref{cond:equiv}, there is a traversal of~$(P_A^{i^*}, P_B^{i^*})$ of width~1.
  We continue with this traversal.
  We continue with~$(\ell_A^{i^*}, \ell_B^{i^*})$ and $(c_A^{i^*}, s_B^{i^*+1})$.
  Then, we traverse~$P_B^{i^*+1}$.
  This is followed by~$(s_A^{i^*+1}, \ell_B^{i^*+1})$ and $(s_A^{i^* + 1}, s_B^{i^* +2})$.
  Thus, we are now at the start of~$P^{i^*+1}_A$ in $P_A$ and~$P_B^{i^*+2}$ in $P_B$.
  
  Afterwards, for $t \ge i > i^*$, we traverse~$P_A^{i}$ and $P_B^{i+1}$ as in the first step:
  First, traverse~$P_A^i$ until the end of~$P_A^i$.
  By \Cref{cond:cB}, this traversal so far has width~1.
  Afterwards, go to~$(\ell_A^{i}, s_B^{i+1})$, followed by~$(c_A^i, s_B^{i+1})$.
  Then traverse~$P_B^{i+1}$.
  By \Cref{cond:cA}, this traversal still has width~1.
  We continue with~$(s_A^{i+1}, \ell_B^{i+1})$ and $(s_A^{i+1}, s_B^{i+2})$.
  
  Finally, we are at $s_A^{t+1}$ and the end of~$P_B^{t+1}$.
  We go to~$(s_A^{t+1}, \ell_B^{t+1})$, followed by~$(s_A^{t+1}, c_B)$ and afterwards traverse~$P_A^{t+1}$.
  By \Cref{cond:cB}, the constructed traversal still has width~1.
  
  Thus, we found a traversal of~$(P_A, P_B)$ of width at most~1.
  It is easy to verify that this traversal has a maximum shift of~$2 \max_{i\in [t+1]} n_i^{\text{DFD}}$.
\end{claimproof}

\begin{claim}\label{claim:rueck}
 If there is no~$i^* \in [t]$ such that there are $a \in A_{i^*}$ and $b \in B_{i^*}$ which are orthogonal, then the Fr\'echet distance of~$P_A$ and $P_B$ is larger than 1.
\end{claim}

\begin{claimproof}
	Assume towards a contradiction that there is a traversal~$T$ of width 1.
	First, we show that $T$ does not reach the first point~$s_B^i$ of~$P_B^i$ for each~$i \in [t + 1]$ before it reaches the first point~$s_A^i$ of~$P_A^i$.
	
	Assume towards a contradiction that~$s_B^i$ is reached before~$s_A^i$ and that~$i$ is minimal.
	Clearly~${i > 1}$ as~$T$ starts with~$(s_A^1,s_B^1)$.
	By choice of~$i$ we have that~$s_A^{i-1}$ is reached before or at the same time as~$s_B^{i-1}$.
	Note that before~$s_B^i$ is reached (which is before~$s_A^i$ is reached) $\ell_B^{i-1}$ has to be reached.
	By \Cref{cond:lB}, the only points from~$P_A$ at distance at most one to~$\ell_B^{i-1}$ are~$\ell_A^j$ and~$s_A^j$ for~$j \in [t]$.
	By choice of~$i$ it follows that~$(s_A^{i-1},\ell_B^{i-1})$ or~$(\ell_A^{i-1},\ell_B^{i-1})$ is part of~$T$.
	We show contradictions in both cases:
	Since~$s_A^{i-1}$ is reached before or simultaneously to~$s_B^{i-1}$, it follows from \Cref{cond:no} that~$(s_A^{i-1},\ell_B^{i-1})$ is not part of~$T$.
	Note that~$(\ell_A^{i-1},\ell_B^{i-1})$ cannot be part of~$T$ as the immediate predecessor consists of the respective last points in~$P_{A_i}$ and~$P_{B_i}$:
	As~$s_A^{i-1}$ is reached before~$s_B^{i-1}$, this would imply~$P_{A_i}$ and~$P_{B_i}$ having a Fr\'echet distance of at most~1, contradicting \Cref{cond:equiv} and~$P_{A_i}$ and~$P_{B_i}$ containing no orthogonal vector.

	Thus, for each~$j \in [t + 1]$ we have that $T$ reaches~$s_A^{j}$ before $s_B^{j}$.
	\Cref{cond:equiv} implies that there is no traversal of width~1 that reaches the last point of~$P_B^{t+1}$ as well as the last point of $P_A^{t+1}$.
	Thus, $T$ must reach the last point of $P_B^{t+1}$ while only partially traversing $P_A^{t+1}$.
	Afterwards, $T$ must move to $\ell_B^{t+1}$, implying that $T$ is at $s_A^{t+1}$ at this point (as all other points have distance larger than 1 to~$\ell_B^{t+1}$).
	However, by \Cref{cond:no}, there is some point in $P_B^{t+1}$ of distance larger than 1 to~$s_A^{t+1}$, a contradiction to $T$ traversing through~$P_B^{t+1}$ while staying at~$s_A^{t+1}$.
\end{claimproof}

The OR-cross-composition can clearly be computed in $\widetilde O(t \max_{i \in [t]} n_i )$ time.
The correctness follows from \Cref{claim:hin,claim:rueck}.
The maximum shift of the constructed instance is $O(n) $ by \Cref{claim:hin}.
Consequently, it is an $(1,1)$-OR-composition.
\end{proof}
}

Combining \Cref{prop:apply-framework} and \Cref{lem:ov-into-dft}, yields the following:

\begin{corollary}\label{cor:DFDlower}
	Unless the SETH fails, there is no $\widetilde O(n^\gamma + \parameter^\beta)$-time algorithm \textsc{Discrete Fr\'echet Distance} parameterized by maximum shift~$\parameter$ for $1 < \gamma < 2$ and $ \beta < \frac{\gamma}{ \gamma - 1}$.
 
	Unless the SETH fails, there is no $\widetilde O(n^\gamma)$-time, $\widetilde O(\parameter^\beta)$-size kernel for \textsc{Discrete Fr\'echet Distance} parameterized by maximum shift for $1 < \gamma < 2$ and $ \beta < \frac{\gamma}{2 \cdot (\gamma - 1)}$.
\end{corollary}

It is easy to extend the $O(n^2)$-algorithm for \textsc{Discrete Fr\'echet Distance}~\cite{EiterM94} to run in~$O(n \parameter)$ time, where~$\parameter$ is the minimal (over all solutions) maximum shift of an optimal solution~\cite{DBLP:conf/spire/Barbay18}.
This together with \Cref{thm:upper-bound} shows that the running time lower bound from \cref{cor:DFDlower} is tight.

We now give lower bounds for another problem from computational geometry, \textsc{Planar Motion Planning}, based not on the hardness of \textsc{Orthogonal Vectors} but on the hardness of \textsc{3-Sum}.
\textsc{Planar Motion Planning} can be solved in $ O(n^2 )$ time~\cite{DBLP:conf/swat/Vegter90}.
Assuming the \textsc{3-Sum} conjecture, \textsc{Planar Motion Planning} cannot be solved in $O(n^{2-\epsilon})$ for any~${\epsilon > 0}$~\cite{DBLP:journals/comgeo/GajentaanO95}.
We say a segment is in the \emph{vicinity} of another segment if they have distance at most the length of the rod.

\begin{lemma}\label{lem:pmp-comp}
 \textsc{3-Sum} $(1,1)$-OR-cross-composes into \textsc{Planar Motion Planning} parameterized by the maximum number of segments any segment has in its vicinity.
\end{lemma}
\begin{proof}[Proof sketch]
 Let~$I_1^{\text{3-Sum}}, \dots, I_t^{\text{3-Sum}}$ be instances of \textsc{3-Sum}.
 We denote by $n_i$ the number of numbers of~$I_1^{\text{3-Sum}}$.
 Gajentaan and Overmars~\cite{DBLP:journals/comgeo/GajentaanO95} gave a reduction from \textsc{3-Sum} to \textsc{Planar Motion Planning} which, given an instance of \textsc{3-Sum} with $n$ numbers, constructs in~$\widetilde O (n )$ time an instance of \textsc{Planar Motion Planning} where the rod initially is in a large ``upper'' rectangle and has to reach a ``lower'' rectangle through a narrow passage (see \Cref{fig:pmp} for a proof by picture).
 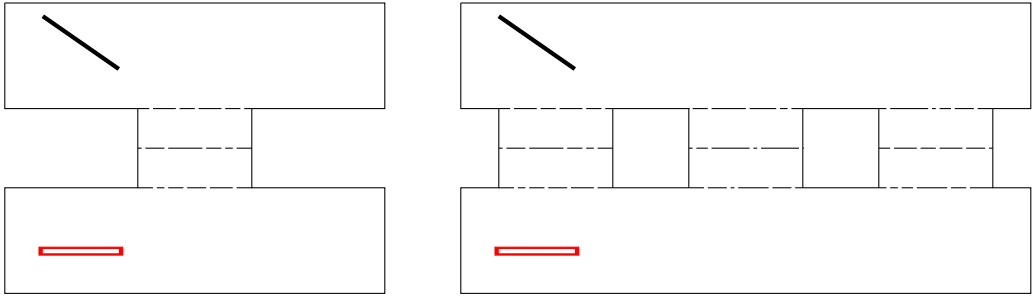
\begin{figure}
  \begin{center}
   \begin{tikzpicture}[yscale=0.35,xscale=0.5]
    \draw (0, 0) -- (10, 0) -- (10, -4) -- (6.5, -4);
    \draw (3.5, -4) -- (0, -4) -- (0, 0);
    
    \draw (3.5, -4) -- (3.5, -7) -- (0, -7) -- (0, -11) -- (10, -11) -- (10, -7) -- (6.5, -7) -- (6.5, -4);
    
    \draw[ultra thick] (1, -0.5) -- (3,-2.5);
    
    \draw (3.5, -4) -- (3.8, -4);
    \draw (3.9, -4) -- (4.5, -4);
    \draw (4.6, -4) -- (5, -4);
    \draw (5.1, -4) -- (5.7, -4);
    \draw (5.8, -4) -- (6.2, -4);
    \draw (6.3, -4) -- (6.5, -4);
    
    \begin{scope}[yshift = -1.5cm]
    \draw (3.5, -4) -- (3.6, -4);
    \draw (3.7, -4) -- (4.2, -4);
    \draw (4.3, -4) -- (5.2, -4);
    \draw (5.3, -4) -- (5.7, -4);
    \draw (5.8, -4) -- (6., -4);
    \draw (6.1, -4) -- (6.5, -4);
    \end{scope}
    
    \begin{scope}[yshift = -3cm]
    \draw (3.5, -4) -- (3.9, -4);
    \draw (4, -4) -- (4.2, -4);
    \draw (4.3, -4) -- (4.7, -4);
    \draw (4.8, -4) -- (5.3, -4);
    \draw (5.4, -4) -- (6., -4);
    \draw (6.1, -4) -- (6.5, -4);
    \end{scope}
    
    \draw[ fill, color=red] (0.9, -9.25) rectangle (3.1, -9.55);
    \draw[ultra thick, color=white] (1, -9.4) -- (3,-9.4);
    
    \begin{scope}[xshift=12cm]
         \draw (0, 0) -- (15, 0) -- (15, -4) -- (14, -4);
		\draw (1, -4) -- (0, -4) -- (0, 0);
		
		\draw (1, -4) -- (1, -7) -- (0, -7) -- (0, -11) -- (15, -11) -- (15, -7) -- (14, -7) -- (14, -4);
		
		\draw (4, -4) -- (6, -4);
		\draw (4, -7) -- (6, -7);
		
		\draw (9, -4) -- (11, -4);
		\draw (9, -7) -- (11, -7);

		\draw (4, -4) -- (4, -7);
		\draw (6, -4) -- (6, -7);
		\draw (9, -4) -- (9, -7);
		\draw (11, -4) -- (11, -7);
		
		\draw[ultra thick] (1, -0.5) -- (3,-2.5);
		
		\begin{scope}[xshift = -2.5cm]
		\draw (3.5, -4) -- (3.8, -4);
		\draw (3.9, -4) -- (4.5, -4);
		\draw (4.6, -4) -- (5, -4);
		\draw (5.1, -4) -- (5.7, -4);
		\draw (5.8, -4) -- (6.2, -4);
		\draw (6.3, -4) -- (6.5, -4);
		
		\begin{scope}[yshift = -1.5cm]
		\draw (3.5, -4) -- (3.6, -4);
		\draw (3.7, -4) -- (4.2, -4);
		\draw (4.3, -4) -- (5.2, -4);
		\draw (5.3, -4) -- (5.7, -4);
		\draw (5.8, -4) -- (6., -4);
		\draw (6.1, -4) -- (6.5, -4);
		\end{scope}
		
		\begin{scope}[yshift = -3cm]
		\draw (3.5, -4) -- (3.9, -4);
		\draw (4, -4) -- (4.2, -4);
		\draw (4.3, -4) -- (4.7, -4);
		\draw (4.8, -4) -- (5.3, -4);
		\draw (5.4, -4) -- (6., -4);
		\draw (6.1, -4) -- (6.5, -4);
		\end{scope}
		\end{scope}

		\begin{scope}[xshift = 2.5cm]
		\draw (3.5, -4) -- (3.5, -4);
		\draw (3.6, -4) -- (4., -4);
		\draw (4.1, -4) -- (5, -4);
		\draw (5.1, -4) -- (5.8, -4);
		\draw (5.9, -4) -- (6.1, -4);
		\draw (6.2, -4) -- (6.5, -4);
		
		\begin{scope}[yshift = -1.5cm]
		\draw (3.5, -4) -- (3.6, -4);
		\draw (3.7, -4) -- (4., -4);
		\draw (4.1, -4) -- (5.1, -4);
		\draw (5.2, -4) -- (5.3, -4);
		\draw (5.4, -4) -- (6.4, -4);
		\draw (6.5, -4) -- (6.5, -4);
		\end{scope}
		
		\begin{scope}[yshift = -3cm]
		\draw (3.5, -4) -- (3.9, -4);
		\draw (4, -4) -- (4.5, -4);
		\draw (4.6, -4) -- (4.7, -4);
		\draw (4.8, -4) -- (5.6, -4);
		\draw (5.7, -4) -- (6.2, -4);
		\draw (6.3, -4) -- (6.5, -4);
		\end{scope}
		\end{scope}
		
		\begin{scope}[xshift = 7.5cm]
		\draw (3.5, -4) -- (3.7, -4);
		\draw (3.8, -4) -- (4.8, -4);
		\draw (4.9, -4) -- (5, -4);
		\draw (5.1, -4) -- (5.4, -4);
		\draw (5.5, -4) -- (6., -4);
		\draw (6.1, -4) -- (6.5, -4);
		
		\begin{scope}[yshift = -1.5cm]
		\draw (3.5, -4) -- (3.8, -4);
		\draw (3.9, -4) -- (4.5, -4);
		\draw (4.6, -4) -- (5., -4);
		\draw (5.1, -4) -- (5.7, -4);
		\draw (5.8, -4) -- (6.3, -4);
		\draw (6.4, -4) -- (6.5, -4);
		\end{scope}
		
		\begin{scope}[yshift = -3cm]
		\draw (3.5, -4) -- (3.7, -4);
		\draw (3.8, -4) -- (4., -4);
		\draw (4.1, -4) -- (4.4, -4);
		\draw (4.5, -4) -- (5.3, -4);
		\draw (5.4, -4) -- (5.8, -4);
		\draw (5.9, -4) -- (6.5, -4);
		\end{scope}
    \end{scope}
    \draw[ fill, color=red] (0.9, -9.25) rectangle (3.1, -9.55);
    \draw[ultra thick, color=white] (1, -9.4) -- (3,-9.4);
    \end{scope}
   \end{tikzpicture}

  \end{center}
  \caption{Left: Exemplary illustration of an instance of \textsc{Planar Motion Planning} constructed by the reduction from \textsc{3-Sum} from Gajentaan and Overmars~\cite{DBLP:journals/comgeo/GajentaanO95}. Right: An example for the OR-cross-composition of three instances of \textsc{3-Sum} into \textsc{Planar Motion Planning}.
  The goal position of the rod is surrounded by red.
  \label{fig:pmp}
  }
 \end{figure}
 We apply this reduction to each instance~$I_i^{\text{3-Sum}}$ of \textsc{3-Sum}, giving us an instance~$I_i^{\text{PMP}}$ with $\widetilde O (n_i)$ many segments.
 From these instances, we construct an instance of \textsc{Planar Motion Planning} as follows:
 We identify the large rectangles in which the rod starts from each~$I_i^{\text{PMP}}$.
 The narrow passages are copied next to each other (if the large starting rectangle is not wide enough, then we make it wider.
 
 The correctness of the reduction is obvious.
 
 Any segment from an instance~$I_i^{\text{3-Sum}}$ has distance at most the length of the rod except for the bounding boxes.
 By splitting the segments of the bounding box into many small segments not longer than the rod, we get that for each segment~$s$ there are at most $O (n)$ segments whose distance to~$s$ is at most the length of the rod.
\end{proof}

Combining \Cref{prop:apply-framework} and \Cref{lem:pmp-comp} yields the following:

\begin{corollary}\label{cor:PMPlower}
 Unless the 3SUM-hypothesis fails, there is no $\widetilde O(n^\gamma + \parameter^\beta)$-time algorithm \textsc{Planar Motion Planning} parameterized by the maximum number~$\parameter$ of segments any segment has in its vicinity for $ \beta < \frac{\gamma}{ \gamma - 1}$.
 
 Unless the 3SUM-hypothesis fails, there is no $\widetilde O(n^\gamma)$-time, $\widetilde O(\parameter^\beta)$-size kernel for \textsc{Planar Motion Planning} parameterized by the maximum number~$\parameter$ of segments any segment has in its vicinity for $ \beta < \frac{\gamma}{2 \cdot (\gamma - 1)}$.
\end{corollary}

The lower bound for the running time is tight by the $ O (K^2 \log n)$ time algorithm from~\citet{DBLP:journals/algorithmica/SifronyS87} (where $K$ is the number of segment pairs whose distance is less than the length of the rod), the observation that $K^2 \le n \cdot \ell$, and \cref{thm:upper-bound}.

\subsection{Graph Problems}
\label{sec:iso}

We now turn to graph problems.
First, we consider \textsc{Minimum Weight $k$-Clique} parameterized by the maximum size of a connected component.

\begin{proposition}\label{prop:min-k-clique}
 \textsc{Minimum Weight $k$-Clique} $(1,1)$-OR-cross-composes into \textsc{Minimum Weight $k$-Clique} parameterized by the maximum size of a connected component.
\end{proposition}

\begin{proof}
 Let $G_1, \dots, G_t$ be instances of \textsc{Minimum Weight $k$-Clique}.
 The cross-composition just computes the disjoint union~$G$ of $G_1, \dots, G_t$.
 
 Clearly, the cross-composition is correct, runs in linear time, and the maximum size of a connected component is bounded by the maximum size of $G_1, \dots, G_s$.
\end{proof}

As \nt\ is the special case of \textsc{Minimum Weight $k$-Clique},
combining \Cref{prop:apply-framework} and \Cref{prop:min-k-clique} yields the following lower bound:

\begin{corollary}\label{cor:ntlower}
 Unless the APSP-hypothesis fails, there is no $\widetilde O(n^\gamma + \parameter^\beta)$-time algorithm for \nt\ parameterized by the maximum size~$\parameter$ of a connected component for $1 < \gamma < 3$ and $ \beta < \frac{2 \cdot \gamma}{ \gamma - 1}$.
 
 Unless the APSP-hypothesis fails, there is no $\widetilde O(n^\gamma)$-time, $\widetilde O(\parameter^\beta)$-size kernel for \nt\ parameterized by the maximum size~$\parameter$ of a connected component for $1 < \gamma < 3$ and $ \beta < \frac{2 \cdot \gamma}{3 \cdot (\gamma - 1)}$.
\end{corollary}
An algorithm running in $O(\ell^2 n)$ where~$\ell$ is the maximum size of a connected component is trivial.
This together with \cref{thm:upper-bound} implies that the running time lower bound is tight.

Next, we turn to \textsc{2nd Shortest Path} which can be solved in $\widetilde O(mn)$~\cite{MR292489} or in $O (M n^\omega)$ time (where $M$ is the largest edge weight)~\cite{DBLP:journals/talg/0001W20}.
If the graph is undirected~\cite{DBLP:journals/networks/KatohIM82} or one aims to find a 2nd shortest walk~\cite{DBLP:journals/siamcomp/Eppstein98}, then there is a quasi-linear-time algorithm.
For unweighted directed graphs, the problem can be solved in $\widetilde O (m \sqrt{n})$ time~\cite{DBLP:journals/talg/RodittyZ12}.
An $\epsilon$-approximation can be computed in $\widetilde O (\frac{m}{\epsilon})$ time~\cite{DBLP:conf/soda/Bernstein10}.

\begin{lemma}\label{lem:2sp}
 \nt\ $(1, 1)$-OR-cross-composes into \textsc{2nd Shortest Path} parameterized by directed feedback vertex number.
\end{lemma}

\begin{proof}
 Let $(I^{\text{NT}}_1, n_1), \dots, (I^{\text{NT}}_t, n_t)$ be instances of \nt.
 We assume without loss of generality that each instance of \nt\ has the same number~$n$ of vertices, i.e., $n_i = n $ for all~$i \in [t]$, and that the largest absolute value of an edge weight is the same for all instances.
 \citet{DBLP:journals/jacm/WilliamsW18} gave a linear-time reduction from \nt\ to \textsc{2nd Shortest Path} which, given an instance~$(I^{\text{NT}}_i = (G_i, w_i), n)$ of \nt, creates an instance~$I^{\text{2SP}}_i$ of \textsc{2nd Shortest Path} as follows:
 For each vertex~$v \in V(G_i)$, add three vertices~$a^v, b^v$, and $c^v$.
 Further, add two vertices~$s^i$ and~$t^i$ as well as a path $s^i, p_1^i, p_2^i, \dots, p_n^i, t^i$, all of whose edges have weight 0.
 Further, there are edges~$(p_i, a^v)$, $(a^v, b^u)$, $(b^u, c^v)$, and $(c^v, p_i)$ for every~$i\in [n], u, v \in V(G_i)$.
 All these edges have positive weight (depending on the edge-weights in~$E(G_i)$). 
 Identifying $s^i$ with $s^j$, $t^i$ with $t^j$, and $p^i_\ell$ with $p^j_\ell$ for all $i,j \in [t]$ and $j \in [\ell]$ then results in one instance of \textsc{2nd Shortest Path} with a directed feedback vertex set of size~$n$, namely $\{p_1, \dots, p_n\}$.
 As the constructed instance is equivalent to the disjunction of~$I^{\text{2SP}}_1$, \dots, $I^{\text{2SP}}_t$ (it is never beneficial to leave $s, p_1, \dots, p_n, t$ more than once as all edges leaving this path have positive weight), we have a $(1, 1)$-OR-cross-composition.
\end{proof}

Combining \Cref{prop:apply-framework} with \Cref{lem:2sp} yields the following:

\begin{corollary}\label{cor:nt-2SPlower}
	Unless the APSP-hypothesis fails, there is no $\widetilde O(n^\gamma + \parameter^\beta)$-time algorithm for \textsc{2nd Shortest Path} parameterized by directed feedback vertex set~$\parameter$ parameterized by the maximum size~$\parameter$ of a connected component for $1 < \gamma < 3$ and $ \beta < \frac{2\gamma}{ \gamma - 1}$.
 
	Unless the APSP-hypothesis fails, there is no $\widetilde O(n^\gamma)$-time, $\widetilde O(\parameter^\beta)$-size kernel for \textsc{2nd Shortest Path} parameterized by directed feedback vertex set parameterized by the maximum size~$\parameter$ of a connected component for $1 < \gamma < 3$ and $ \beta < \frac{2\gamma}{3(\gamma - 1)}$.
\end{corollary}

In contrast to the other problems studied in this paper, we do not know whether the running time lower bounds for \textsc{2nd Shortest Path} are tight.

\subsection{Triangle Collection}

Last, we consider the \textsc{Triangle Collection} problem.
\textsc{Triangle Collection} can trivially be solved in $O(n^3)$ time, but does not admit an $O(n^{3-\epsilon})$ time algorithm assuming SETH, the \textsc{3-Sum} conjecture, \emph{or} the \textsc{APSP} conjecture~\cite{DBLP:journals/siamcomp/AbboudWY18}.
For this problem, we were unable to apply our framework directly, i.e., to find an OR-cross-composition from an OR-decomposable problem to \textsc{Triangle Collection}.
However, we can still get a lower bound in a very similar fashion by combining decomposition and composition into one step.
The difference is that the decomposition of \textsc{Triangle Collection} that we use in the proof of the following result does not admit the ``OR''-property. 

\begin{proposition}\label{prop:tclower}
	Unless \OV, \textsc{3-Sum}, or \textsc{APSP}-hypothesis fails, there is no $\widetilde O(n^\gamma + \parameter^\beta)$-time algorithm for \textsc{Negative Triangle} parameterized by the maximum size~$\parameter$ of a connected component for $1 < \gamma < 3$ and $ \beta < \frac{2 \cdot \gamma}{ \gamma - 1}$.
 
	Unless \OV, \textsc{3-Sum}, or \textsc{APSP}-hypothesis fails, there is no $\widetilde O(n^\gamma)$-time, $\widetilde O(\parameter^\beta)$-size kernel for \textsc{Triangle Collection} parameterized by the maximum size~$\parameter$ of a connected component for $1 < \gamma < 3$ and $ \beta < \frac{2 \cdot \gamma}{3 \cdot (\gamma - 1)}$.
\end{proposition}

\begin{proof}
 Fix $1 < \gamma < 3$.
 Let~$G$ be in an instance of \textsc{Triangle Collection}.
 Partition $V(G)$ into $z \coloneqq n^{\nicefrac{\lambda}{3 + \lambda}}$ sets~$V_1, \dots, V_z$ of size $q \coloneqq n^{\nicefrac{3}{(3+ \lambda)}}$, where $\lambda \coloneqq \nicefrac{\beta}{\gamma} -1$.
 For each~$(i,j, k ) \in [z]^3$, let~$G_{(i,j,k)}$ be the graph induced by~$V_i \cup V_j \cup V_k$.
 Finally, let~$H$ be the union of all~$G_{(i,j,k)}$.
 Note that $H$ corresponds to the output of the OR-decomposition in our framework.
 Clearly $G$ has a triangle collection if and only if $H$ has.
 Further, $H$ can be computed in $\widetilde O (q^2 \cdot z^3 ) = \widetilde O (n^{\nicefrac{3(\lambda + 2)}{(3 + \lambda)}})$ time (note that each of the $z^3$ graphs~$G_{(i,j,k)}$ has size $O(q^2)$ as a $q$-vertex graph may have $O(q^2)$ many edges).
 
 Now assume that there is a $\widetilde O (n^{\gamma} + \ell^\beta)$-algorithm for \textsc{Triangle Collection} with $\beta < \frac{2 \cdot \gamma}{ \gamma - 1}$ and apply this algorithm to~$H$.
 This clearly solves~$H$.
 The running time for this step is $\widetilde O((z^3\cdot q)^{\gamma} + q^{\beta}) = \widetilde O (n^{ \gamma\cdot (\nicefrac{3 \cdot (\lambda + 1)}{(3 + \lambda)})} + n^{\beta \cdot \nicefrac{3}{(3+ \lambda)}})$ time.
 Note that $\nicefrac{\gamma\cdot 3\cdot (\lambda + 1)}{(3+\lambda)} = 3 \cdot \frac{ \nicefrac{\beta}{\gamma}}{2 + \nicefrac{\beta}{\gamma}} < 3$.
 Further, it holds that 
 \[
  \beta \cdot \nicefrac{3}{(3+ \lambda)} = 3 \cdot \frac{\beta}{2+ \nicefrac{\beta}{\gamma}} = 3 \cdot \frac{\beta \gamma}{2\gamma + \beta} = 3 \cdot \frac{\gamma}{\nicefrac{2\gamma}{\beta} + 1} < 3 \cdot \frac{\gamma}{\nicefrac{2 \gamma}{\bigl(\frac{2\gamma}{(\gamma -1)}\bigr)} + 1}= 3 \cdot \frac{\gamma}{(\gamma - 1) + 1} = 3
 \]
 where we used the assumption $\beta < \frac{2 \cdot \gamma}{ \gamma - 1}$ for the inequality.
 Consequently, we can solve \textsc{Triangle Collection} in $O (n^{3-\epsilon})$ time for some  $\epsilon > 0$.
 
 The statement for kernels follows from the observation that an $O(\ell^\beta)$-size, $O(n^{\gamma})$-time kernel for \textsc{Triangle Collection} directly implies an $O(n^{\gamma} + \ell^{3 \cdot \beta})$-time algorithm for \textsc{Triangle Collection}.
\end{proof}

As an $O(n \ell^2)$ time algorithm for \textsc{Triangle Collection} is trivial, it follows from \cref{thm:upper-bound} that the running time lower bound is tight.

\section{Conclusion}

We introduced a framework for extending conditional running time lower bounds to parameterized running time lower bounds and applied it to various problems.
Beyond the clear task to apply the framework to further problems there are further challenges for future work.
For example, can we get ``AND-hard'' problems so that we can use AND-cross compositions similar to the ones used to exclude compression~\cite{Dru15}?
Moreover, can the framework be adapted to cope with dynamic, counting, or enumerating problems?

\bibliographystyle{plainnat}
\bibliography{bib}

\end{document}